\newif\ifarxiv
\arxivtrue

\documentclass[a4paper,UKenglish,cleveref,autoref,thm-restate]{lipics-v2021}

\ifarxiv
\pdfoutput=1
\nolinenumbers
\usepackage[T1]{fontenc}
\usepackage[scale=0.92]{inconsolata}
\else
\usepackage{fontspec}
\fi

\usepackage[supertabular]{ottalt}

\usepackage{enumitem,booktabs,xspace,doi}
\usepackage{mathpartir,mathtools,stmaryrd}
\usepackage[bottom,flushmargin,multiple,para]{footmisc} 

\newcommand{\citep}[1]{\cite{#1}}
\newcommand{\citet}[1]{\cite{#1}}
\newcommand{\repo}{https://github.com/ionathanch/TTBFL}
\newcommand{\lang}{TTBFL\@\xspace}
\newcommand{\titlebreak}{\texorpdfstring{\\}{}}
\newcommand{\ie}{\textit{i.e.}\@\xspace}
\newcommand{\eg}{\textit{e.g.}\@\xspace}

\newcommand{\ala}{\textit{\`a la}\@\xspace}
\newcommand{\apriori}{\textit{a priori}\@\xspace}

\newcommand{\thmref}[2]{%
  $\langle$\textnormal{\texttt{\href{\repo/tree/main/src/#1}{#1}:#2}}$\rangle$%
}

\hypersetup{
  colorlinks=true,
  urlcolor=blue,
  linkcolor=magenta,
  citecolor=teal
}
\title{Bounded First-Class Universe Levels \titlebreak in Dependent Type Theory}
\titlerunning{Bounded First-Class Universe Levels}
\authorrunning{J. Chan, S. Weirich}
\Copyright{Jonathan Chan, Stephanie Weirich}
\ccsdesc{Theory of computation~Type theory}
\keywords{type theory, universes, universe polymorphism}
\hideLIPIcs

\author{Jonathan Chan}
  {University of Pennsylvania, Philadelphia, USA}
  {jcxz@seas.upenn.edu}
  {0000-0003-0830-3180}
  {}

\author{Stephanie Weirich}
  {University of Pennsylvania, Philadelphia, USA}
  {sweirich@seas.upenn.edu}
  {0000-0002-6756-9168}
  {}

\supplementdetails[subcategory={source code},
  swhid={swh:1:dir:8f18b01234056282a037b3d835e97df2b5050b29} 
]{Software}{\repo}

\newcommand{\ottnt}[1]{\mathit{#1}}
\newcommand{\ottmv}[1]{\mathit{#1}}

\newcommand{\ottdrulename}[1]{\textsc{#1}}

\newcommand{\ottafterlastrule}{\\}
    \newcommand{\gap}{\:}
    \newcommand{\kw}[1]{ \mathsf{#1} }

  \renewottcommands[ott]

\begin{document}

\setlength{\abovedisplayskip}{0.25\baselineskip}
\setlength{\belowdisplayskip}{0.25\baselineskip}

\maketitle

\begin{abstract}
  In dependent type theory,
  being able to refer to a type universe as a term itself increases its expressive power,
  but requires mechanisms in place to prevent Girard's paradox
  from introducing logical inconsistency in the presence of type-in-type.
  The simplest mechanism is a hierarchy of universes indexed by a sequence of levels,
  typically the naturals.
  To improve reusability of definitions, they can be made level polymorphic,
  abstracting over level variables and adding a notion of level expressions.
  For even more expressive power,
  level expressions can be made first-class as terms themselves,
  and level polymorphism is subsumed by dependent functions quantifying over levels.
  Furthermore, bounded level polymorphism provides more expressivity
  by being able to explicitly state constraints on level variables.
  While semantics for first-class levels with constraints are known,
  syntax and typing rules have not been explicitly written down.
  Yet pinning down a well-behaved syntax is not trivial;
  there exist prior type theories with bounded level polymorphism
  that fail to satisfy subject reduction.
  In this work, we design an explicit syntax for
  a type theory with bounded first-class levels,
  parametrized over arbitrary well-founded sets of levels.
  We prove the metatheoretic properties of subject reduction,
  type safety, consistency, and canonicity,
  entirely mechanized from syntax to semantics in Lean.
\end{abstract}

\section{Introduction}

Dependent type theories are common foundations for proof assistants,
where theorems are manipulated as types and their proofs as terms.
Types are often treated as terms themselves,
providing a uniform mechanism for working with both;
for example, quantifying over predicates is no different from quantifying over functions,
as predicates are functions that return types.
To merge types and terms, we need a type of types, or a \emph{universe},
which itself must be a term with a type.

Girard~\citep{systemf} showed that a type-in-type axiom makes dependent type theory logically inconsistent:
if the type of a universe is itself, then all types are inhabited,
rendering the type theory useless as a tool for proving.
Therefore, Martin-L\"of stratified the universe in his type theory (MLTT)~\citep{mltt}
into a countably infinite hierarchy of universes
$ \kw{U}_{  0  }  :  \kw{U}_{  1  }  :  \kw{U}_{  2  }  : \dots$
indexed by \emph{universe levels} spanning the naturals.
Many contemporary proof assistants based on dependent types feature such a hierarchy,
such as Rocq~\citep{rocq}, Agda~\citep{agda}, Lean~\citep{lean}, and F$^\star$~\citep{fstar}.

Having only a concrete universe hierarchy, however,
limits the reusability of definitions that are not inherently tied to particular universe levels.
For example, the identity function $\kw{id} :   \Pi  A  \mathbin{:}   \kw{U}_{  \mathit{ i }  }   \mathpunct{.}   \mathit{ A }    \to   \mathit{ A }  $
would need to be redefined for each universe level $\ottmv{i}$ at which it is needed.
Universe level polymorphism addresses this issue by abstracting over level variables,
used to index universes alongside concrete levels.
Its simplest form is prenex level polymorphism,
introduced by Harper and Pollack~\citep{anon-univ},
which restricts the abstraction to top-level definitions.
Courant~\citep{explicit} extends their implicit system to an explicit system
with (in)equality constraints, level operators, and level expressions.
This extension is implemented in Rocq~\citep{univ-poly-coq}.

If we disallow recursive definitions that vary in the level,
uses of prenex-polymorphic definitions
can be specialized to level-monomorphic terms.
Favonia, Angiuli, and Mullanix note that it
``is as consistent as standard (monomorphic) type theory [\dots]
because any given proof can only mention finitely many universes'',
and show consistency using this idea~\citep{displacement}.

If level quantification is added as a type former directly to the type theory,
we obtain higher-rank level polymorphism,
where level-polymorphic terms can be passed as arguments to functions.
For instance, Bezem, Coquand, Dybjer, and Escard\'o
introduce such a type theory (referred to here as BCDE)
with level constraints~\citep{univ-poly}.
Going further, rather than keeping universe levels distinct from terms,
we can make them first class by defining level expressions as a subset of terms,
and add a type of levels;
such levels are found in Agda.
Level quantification is subsumed by dependent functions whose domain is this level type.
The codomain can also be the level type,
which describes functions that compute levels.

First-class universe levels are known to be logically consistent.
In particular, Kov\'{a}cs~\citep{gen-univ} gives a semantic model for a type theory TTFL,
which features first-class levels and an ordering relation $<$ on them.
The model is given as categories with families (cwfs)~\citep{cwf},
mostly mechanized in Agda using induction--recursion,
and supports features such as level constraints,
maxima of levels, and induction on levels.

The syntax of TTFL is considered to be the initial model in the category of cwfs,
but an explicit syntax and typing rules are not given,
and proving initiality even for MLTT is a colossal task~\citep{initiality}.
Furthermore, while a syntax may satisfy semantic properties such as logical consistency,
it may not necessarily satisfy desirable syntactic properties.
In particular, BCDE's semantics can conceivably be viewed as
that of TTFL without making levels first class,
yet its syntax fails to satisfy subject reduction.

In this work, we give an end-to-end account of first-class levels in type theory,
beginning with an explicit syntax and typing rules,
and proving that they satisfy desirable metatheoretic properties.
Our contributions are as follows:

\begin{itemize}[topsep=0pt]
  \item We present \textbf{\lang},
    a dependent type theory with bounded, first-class universe levels.
    Our bounds differ from level constraints in that
    they are inherent to the type of a level,
    rather than a separate predicate on them,
    which prevents failure of subject reduction.
    Examples in the next section build up from monomorphic levels to level polymorphism
    before we proceed to the formal definition of the type theory in \cref{sec:ttbfl}.
  \item We prove subject reduction (\ie preservation) in \cref{sec:safety},
    an improvement upon the metatheoretic properties of BCDE.
    We also prove progress and thus type safety,
    which is important if we also want to use the language for writing programs that evaluate.
    An example is implementing proof assistants in themselves,
    as is (partially) done in Lean and undergoing work for Rocq~\citep{coq-in-coq}.
  \item Using a syntactic logical relation,
    we prove logical consistency and canonicity
    via the fundamental soundness theorem in \cref{sec:lr}.
    Consistency ensures that the type theory is suitable
    as a basis for logical reasoning in a proof assistant,
    while canonicity ensures that closed terms evaluate to the values we expect.
    Normalization of open terms remains an open problem (\cref{sec:normalization}).
\end{itemize}
All results are mechanized in Lean.
The development consists of under 1600 lines of code,
which can be found in the supplementary materials at \url{\repo}.
The definitions and theorems in this paper
are hyperlinked to the corresponding Lean files.

As our system is intentionally very minimal,
we discuss some further extensions in \cref{sec:extensions},
including level operators and subtyping.
We conclude with future work in \cref{sec:conclusion}.

\section{Motivation}

To motivate the range of features in \lang,
we look at examples starting from monomorphic universe levels
and build up to first-class levels and bounding in this section.
Although not found in our minimal language,
these examples use dependent pairs, propositional equality, the naturals, and lists
for more illuminating examples.

Let us start by revisiting the identity function and its type,
supposing $ \kw{U}_{  0  }  :  \kw{U}_{  1  }  : \dots  \kw{U}_{  \omega  } $,
with a limit universe $ \omega $ at the top,
which will come in handy later.
\begin{align*}
   \kw{ Id }  &:  \kw{U}_{  1  } 
  & \kw{ id }  &:  \kw{ Id }  \\
   \kw{ Id }  &\coloneqq   \Pi  A  \mathbin{:}   \kw{U}_{  0  }   \mathpunct{.}   \mathit{ A }    \to   \mathit{ A }  
  &  \kw{ id }  &\coloneqq  \lambda  A  \mathbin{:}   \kw{U}_{  0  }   \mathpunct{.}   \lambda  x  \mathbin{:}   \mathit{ A }   \mathpunct{.}   \mathit{ x }   
\end{align*}

This identity function is polymorphic over types in $ \kw{U}_{  0  } $,
but not over universes, so the self application
$   \kw{ id }   \gap   \kw{ Id }    \gap   \kw{ id }  $ is ill typed.
More generally, if we want to reuse a definition at different universe levels,
it would need to be redefined for every level needed.
If we introduce prenex polymorphism of universe levels,
where top-level definitions are permitted to be polymorphic,
we can write a universe polymorphic identity function
that can be instantiated at different levels and self-applied.
\begin{align*}
   \kw{ Id }  &:  \forall  i  \mathpunct{.}   \kw{U}_{    \mathit{ i }   + 1   }  
  & \kw{ id }  &:   \forall  i  \mathpunct{.}   \kw{ Id }    \gap [   \mathit{ i }   ]  \\
   \kw{ Id }  &\coloneqq   \Lambda  i  \mathpunct{.}   \Pi  A  \mathbin{:}   \kw{U}_{  \mathit{ i }  }   \mathpunct{.}   \mathit{ A }     \to   \mathit{ A }  
  & \kw{ id }  &\coloneqq  \Lambda  i  \mathpunct{.}   \lambda  A  \mathbin{:}   \kw{U}_{  \mathit{ i }  }   \mathpunct{.}   \lambda  x  \mathbin{:}   \mathit{ A }   \mathpunct{.}   \mathit{ x }    
\end{align*}

Now, the expression $    \kw{ id }   \gap [   1   ]   \gap   (   \kw{ Id }   \gap [   0   ]  )    \gap   (   \kw{ id }   \gap [   0   ]  )  $ is well typed.
A definition can also be polymorphic over multiple levels,
such as the constant function that takes two arguments but always returns the first.
For this, we need a binary least upper bound operator $  \relax   \sqcup   \relax  $ on levels.
\begin{align*}
   \kw{ Const }  &:  \forall  i  \mathpunct{.}   \forall  j  \mathpunct{.}   \kw{U}_{    (   \mathit{ i }   \sqcup   \mathit{ j }   )   + 1   }   
  & \kw{ const }  &:    \forall  i  \mathpunct{.}   \forall  j  \mathpunct{.}   \kw{ Const }     \gap [   \mathit{ i }   ]   \gap [   \mathit{ j }   ]  \\
   \kw{ Const }  &\coloneqq    \Lambda  i  \mathpunct{.}   \Lambda  j  \mathpunct{.}   \Pi  A  \mathbin{:}   \kw{U}_{  \mathit{ i }  }   \mathpunct{.}   \Pi  B  \mathbin{:}   \kw{U}_{  \mathit{ j }  }   \mathpunct{.}   \mathit{ A }       \to   \mathit{ B }    \to   \mathit{ A }  
  & \kw{ const }  &\coloneqq  \Lambda  i  \mathpunct{.}   \Lambda  j  \mathpunct{.}   \lambda  A  \mathpunct{.}   \lambda  B  \mathpunct{.}   \lambda  x  \mathpunct{.}   \lambda  y  \mathpunct{.}   \mathit{ x }       
\end{align*}

The universe in which $   \kw{ Const }   \gap [   \ottmv{i}   ]   \gap [   \ottmv{j}   ] $ lives is $  (   \ottmv{i}   \sqcup   \ottmv{j}   )   + 1 $,
because its universe must contain the universes
$ \kw{U}_{  \ottmv{i}  } $ and $ \kw{U}_{  \ottmv{j}  } $ over which it quantifies.
As more level variables get involved,
the algebraic expressions on levels becomes increasingly complex.
But the precise universe in which $ \kw{ Const } $ lives is not as important
as knowing that it lives in \emph{some} greater universe,
which is all that is needed to prevent type-in-type inconsistencies.
This can be expressed by bounded level quantification,
simplifying level expressions at the cost of an additional level variable.
We use the limit level $ \omega $ to allow $\ottnt{k}$ to range over all other levels.
\begin{align*}
   \kw{ Const }  &:  \forall  k  <   \omega   \mathpunct{.}   \forall  i  <   \mathit{ k }   \mathpunct{.}   \forall  j  <   \mathit{ k }   \mathpunct{.}   \kw{U}_{  \mathit{ k }  }     \\
   \kw{ Const }  &\coloneqq    \Lambda  k  \mathpunct{.}   \Lambda  i  \mathpunct{.}   \Lambda  j  \mathpunct{.}   \Pi  A  \mathbin{:}   \kw{U}_{  \mathit{ i }  }   \mathpunct{.}   \Pi  B  \mathbin{:}   \kw{U}_{  \mathit{ j }  }   \mathpunct{.}   \mathit{ A }        \to   \mathit{ B }    \to   \mathit{ A }  
\end{align*}

While nonrecursive prenex level polymorphism can be monomorphized away,
this is not the case once we introduce recursive definitions
whose recursive calls may vary in the level.
This lets us define universes with levels incremented by fixed amount,
\ie $ \kw{U}_{    \mathit{ k }   +   \mathit{ n }    } $.
\begin{align*}
  & \kw{ incr }  :   \forall  k  <   \omega   \mathpunct{.}   \kw{ Nat }    \to   \kw{U}_{  \omega  }   \\
  &   \kw{ incr }   \gap  \ottnt{k}   \gap   \kw{ zero }   \coloneqq  \kw{U}_{  \mathit{ k }  }  \\
  &   \kw{ incr }   \gap  \ottnt{k}   \gap   (   \kw{ succ }   \gap  \ottnt{n}  )   \coloneqq    \kw{ incr }   \gap  \ottnt{n}   \gap [   \ottnt{k}  + 1   ] 
\end{align*}

Generalizing from prenex level polymorphism to higher-rank level polymorphism affords even more reusability.
One application is when axioms are explicitly assumed
as local hypotheses instead of globally axiomatized
to restrict their usage to only where they are really needed.
An example is function extensionality,
whose type is level polymorphic.
\begin{align*}
   \kw{ FunExt }  &:  \forall  k  <   \omega   \mathpunct{.}   \forall  i  <   \mathit{ k }   \mathpunct{.}   \forall  j  <   \mathit{ k }   \mathpunct{.}   \kw{U}_{  \mathit{ k }  }     \\
   \kw{ FunExt }  &\coloneqq  \Lambda  k  \mathpunct{.}   \Lambda  i  \mathpunct{.}   \Lambda  j  \mathpunct{.}   \Pi  A  \mathbin{:}   \kw{U}_{  \mathit{ i }  }   \mathpunct{.}   \Pi  B  \mathbin{:}   (   \mathit{ A }   \to   \kw{U}_{  \mathit{ j }  }   )   \mathpunct{.}   \relax       \\
  &\phantom{{} \coloneqq {}}   \Pi  f  \mathbin{:}   (   \Pi  x  \mathbin{:}   \mathit{ A }   \mathpunct{.}   \mathit{ B }    \gap   \mathit{ x }   )   \mathpunct{.}   \Pi  g  \mathbin{:}   (   \Pi  x  \mathbin{:}   \mathit{ A }   \mathpunct{.}   \mathit{ B }    \gap   \mathit{ x }   )   \mathpunct{.}   (     \Pi  x  \mathbin{:}   \mathit{ A }   \mathpunct{.}   \mathit{ f }    \gap   \mathit{ x }    =   \mathit{ g }    \gap   \mathit{ x }   )     \to     \mathit{ f }   =   \mathit{ \ottmv{g} }    
\end{align*}

Suppose we wished to prove that function extensionality for
functions with two arguments at different universe levels follows from assuming $ \kw{ FunExt } $.
Using only prenex polymorphism, we would need two separate instantiations,
once for its application to the functions of type $ \Pi  \ottmv{x}  \mathbin{:}  \ottnt{A}  \mathpunct{.}     \ottnt{B}  \gap   \mathit{ \ottmv{x} }    \to  \ottnt{C}   \gap   \mathit{ \ottmv{x} }   $,
and once for its application to the functions of type $   \ottnt{B}  \gap   \mathit{ \ottmv{x} }    \to  \ottnt{C}   \gap   \mathit{ \ottmv{x} }  $.
\begin{align*}
   \kw{ lemma }  &:    \forall  l  <   \omega   \mathpunct{.}   \forall  i  <   \mathit{ l }   \mathpunct{.}   \forall  j  <   \mathit{ l }   \mathpunct{.}   \forall  k  <   \mathit{ l }   \mathpunct{.}   (     \kw{ FunExt }   \gap [   \mathit{ l }   ]   \gap [   \mathit{ i }   ]   \gap [    \mathit{ j }   \sqcup   \mathit{ k }    ]  )       \to   (     \kw{ Funext }   \gap [   \mathit{ l }   ]   \gap [   \mathit{ j }   ]   \gap [   \mathit{ k }   ]  )    \to   \relax   \\
  &\phantom{{} : {}}  \Pi  A  \mathbin{:}   \kw{U}_{  \mathit{ i }  }   \mathpunct{.}   \Pi  B  \mathbin{:}   (   \mathit{ A }   \to   \kw{U}_{  \mathit{ j }  }   )   \mathpunct{.}   \Pi  C  \mathbin{:}   (   \mathit{ A }   \to   \kw{U}_{  \mathit{ k }  }   )   \mathpunct{.}   \relax     \\
  &\phantom{{} : {}}  \Pi  f  \mathbin{:}   (     \Pi  x  \mathbin{:}   \mathit{ A }   \mathpunct{.}   \mathit{ B }    \gap   \mathit{ x }    \to   \mathit{ C }    \gap   \mathit{ x }   )   \mathpunct{.}   \Pi  g  \mathbin{:}   (     \Pi  x  \mathbin{:}   \mathit{ A }   \mathpunct{.}   \mathit{ B }    \gap   \mathit{ x }    \to   \mathit{ C }    \gap   \mathit{ x }   )   \mathpunct{.}   \relax    \\
  &\phantom{{} : {}}   (       \Pi  x  \mathbin{:}   \mathit{ A }   \mathpunct{.}   \Pi  y  \mathbin{:}    \mathit{ B }   \gap   \mathit{ x }    \mathpunct{.}   \mathit{ f }     \gap   \mathit{ x }    \gap   \mathit{ y }    =   \mathit{ g }    \gap   \mathit{ x }    \gap   \mathit{ y }   )   \to     \mathit{ f }   =   \mathit{ \ottmv{g} }     \\
   \kw{ lemma }  &\coloneqq  \Lambda  l  \mathpunct{.}   \Lambda  i  \mathpunct{.}   \Lambda  j  \mathpunct{.}   \Lambda  k  \mathpunct{.}   \lambda  fe1  \mathpunct{.}   \lambda  fe2  \mathpunct{.}   \relax        \dots
\end{align*}

Once more universe levels get involved,
instantiating up front every possible use becomes unwieldy.
With higher-rank polymorphism,
we can quantify over a polymorphic function extensionality principle once and for all,
and instantiate its levels within the proof as needed.
\begin{align*}
   \kw{ lemma }  &:   (     \forall  k  <   \omega   \mathpunct{.}   \forall  i  <   \mathit{ k }   \mathpunct{.}   \forall  j  <   \mathit{ k }   \mathpunct{.}   \kw{ FunExt }      \gap [   \mathit{ i }   ]   \gap [   \mathit{ j }   ]   \gap [   \mathit{ k }   ]  )   \to   \relax   \\
  &\phantom{{} : {}}  \forall  l  <   \omega   \mathpunct{.}   \forall  i  <   \mathit{ l }   \mathpunct{.}   \forall  j  <   \mathit{ l }   \mathpunct{.}   \forall  k  <   \mathit{ l }   \mathpunct{.}   \Pi  A  \mathbin{:}   \kw{U}_{  \mathit{ i }  }   \mathpunct{.}   \Pi  B  \mathbin{:}   (   \mathit{ A }   \to   \kw{U}_{  \mathit{ j }  }   )   \mathpunct{.}   \Pi  C  \mathbin{:}   (   \mathit{ A }   \to   \kw{U}_{  \mathit{ k }  }   )   \mathpunct{.}   \relax         \\
  &\phantom{{} : {}}  \Pi  f  \mathbin{:}   (     \Pi  x  \mathbin{:}   \mathit{ A }   \mathpunct{.}   \mathit{ B }    \gap   \mathit{ x }    \to   \mathit{ C }    \gap   \mathit{ x }   )   \mathpunct{.}   \Pi  g  \mathbin{:}   (     \Pi  x  \mathbin{:}   \mathit{ A }   \mathpunct{.}   \mathit{ B }    \gap   \mathit{ x }    \to   \mathit{ C }    \gap   \mathit{ x }   )   \mathpunct{.}   \relax    \\
  &\phantom{{} : {}}   (       \Pi  x  \mathbin{:}   \mathit{ A }   \mathpunct{.}   \Pi  y  \mathbin{:}    \mathit{ B }   \gap   \mathit{ x }    \mathpunct{.}   \mathit{ f }     \gap   \mathit{ x }    \gap   \mathit{ y }    =   \mathit{ g }    \gap   \mathit{ x }    \gap   \mathit{ y }   )   \to     \mathit{ f }   =   \mathit{ \ottmv{g} }     \\
   \kw{ lemma }  &\coloneqq  \lambda  fe  \mathpunct{.}   \Lambda  l  \mathpunct{.}   \Lambda  i  \mathpunct{.}   \Lambda  j  \mathpunct{.}   \Lambda  k  \mathpunct{.}   \relax       \dots
\end{align*}

With higher-rank level polymorphism,
a level-polymorphic type itself must live in some universe,
which is often that of the bounding level.
Coming back to the identity function,
we can impose a bound on its level by bounded quantification,
and use the bound for the universe.
Self-applications such as $     \kw{ id }   \gap [   2   ]   \gap [   1   ]   \gap   (   \kw{ Id }   \gap [   1   ]  )    \gap   (   \kw{ id }   \gap [   1   ]  )  $ still hold.
\begin{align*}
   \kw{ Id }  &:  \forall  j  <   \omega   \mathpunct{.}   \kw{U}_{  \mathit{ j }  }   &
   \kw{ id }  &:   \forall  j  <   \omega   \mathpunct{.}   \kw{ Id }    \gap [   \mathit{ j }   ]  \\
   \kw{ Id }  &\coloneqq   \Lambda  j  \mathpunct{.}   \forall  i  <   \mathit{ j }   \mathpunct{.}   \Pi  A  \mathbin{:}   \kw{U}_{  \mathit{ i }  }   \mathpunct{.}   \mathit{ A }      \to   \mathit{ A }   &
   \kw{ id }  &\coloneqq  \Lambda  j  \mathpunct{.}   \Lambda  i  \mathpunct{.}   \lambda  A  \mathbin{:}   \kw{U}_{  \mathit{ i }  }   \mathpunct{.}   \lambda  x  \mathbin{:}   \mathit{ A }   \mathpunct{.}   \mathit{ x }     
\end{align*}

So far, our notions of level polymorphism treat levels as syntactically separate from terms,
with special level operators $\cdot   \relax   + 1 $ and $\cdot   \relax   \sqcup   \relax   \cdot$.
Consequently, if we want more general ways to compute level expressions,
we must add them as primitives to the language.
If we instead make levels first class,
we are then able to manipulate and store them as terms.
Bounded level quantifications are subsumed by ordinary dependent types
whose domain is the type of all levels bounded by some strictly greater level $ \kw{Level}\texttt{<} \gap  \ottnt{k} $.
An example application is computing the least upper bound level
from a list of levels and types of that level.
\begin{align*}
   \kw{ lub }  &:    \kw{ List }   \gap   (  \Sigma  i  \mathbin{:}   \kw{Level}\texttt{<} \gap   \omega    \mathpunct{.}   \kw{U} \gap   \mathit{ i }    )    \to   \kw{Level}\texttt{<} \gap   \omega    \\
   \kw{ lub }  &\gap  \kw{ nil }  \coloneqq 0 \\
   \kw{ lub }  &\gap  (    \kw{ cons }   \gap   (   \ottmv{i}   ,  \ottnt{A}  )    \gap   \mathit{ As }   )  \coloneqq   \ottmv{i}   \sqcup   (   \kw{ lub }   \gap   \mathit{ As }   )  
\end{align*}

This level computation can be used to turn a list of types and their levels
into an n-ary tuple with a precise level.
This is a technique used, for instance, by Escot and Cockx in generic programming
to represent level-polymorphic inductive types~\citep{generic}.
\begin{align*}
   \kw{ Interp }  &:  \Pi  As  \mathbin{:}    \kw{ List }   \gap   (  \Sigma  i  \mathbin{:}   \kw{Level}\texttt{<} \gap   \omega    \mathpunct{.}   \kw{U} \gap   \mathit{ i }    )    \mathpunct{.}   \kw{U} \gap   (   \kw{ lub }   \gap   \mathit{ As }   )    \\
   \kw{ Interp }  &\gap  \kw{ nil }  \coloneqq  \top  \\
   \kw{ Interp }  &\gap  (    \kw{ cons }   \gap   (   \ottmv{i}   ,  \ottnt{A}  )    \gap   \mathit{ As }   )  \coloneqq  \ottnt{A}  \times   (   \kw{ Interp }   \gap   \mathit{ As }   )  
\end{align*}

Various proof assistants with universe level polymorphism implement different subsets of these features.
Lean and F$^\star$ have prenex polymorphism with successor and least upper bound operators.
Rocq has prenex polymorphism along with level (in)equality declarations,
but no other operators.
Agda has first-class levels and the two level operator, but no level constraints.
In \lang, we include bounded first-class levels,
but omit the two level operators for simplicity,
opting to treat them as straightforward potential extensions.

\section{A minimal type theory with bounded first-class universe levels} \label{sec:ttbfl}

\lang is a Church-style type theory \ala Russell,
where terms may have type annotations,
and there is no separate typing judgement for well-formedness of types.
To keep the type theory minimal, it contains only dependent functions,
an empty type, predicative universes, and bounded universe levels.
By convention, we use $\ottnt{a}, \ottnt{b}, \ottnt{c}$ for terms,
$\ottnt{A}, \ottnt{B}, \ottnt{C}$ for types,
and $\ottnt{k}, \ell$ for level terms.
The syntax is presented in \cref{fig:syntax};
we additionally use $ \ottnt{A}  \to  \ottnt{B} $ as sugar for nondependent functions
$ \Pi  \ottmv{x}  \mathbin{:}  \ottnt{A}  \mathpunct{.}  \ottnt{B} $ where $\ottmv{x}$ does not occur in $\ottnt{B}$.
While the mechanization uses de Bruijn indexing and simultaneous substitutions,
this paper presents the syntax in nominal form for clarity,
and we omit the details of manipulating substitutions for concision.
We write single substitutions of a variable $\ottmv{x}$
in a term $\ottnt{b}$ by another term $\ottnt{a}$ as $ \ottnt{b} [  \ottmv{x}  \mapsto  \ottnt{a}  ] $.

\begin{figure}
\vspace{-\baselineskip}
\begin{align*}
  i, j & \Coloneqq \texttt{<concrete universe levels>} \\
  x, y, z & \Coloneqq \texttt{<term variables>} \\
  a, b, c, A, B, C, k, \ell & \Coloneqq \ottmv{x} \mid \ottmv{i}
    \mid  \Pi  \ottmv{x}  \mathbin{:}  \ottnt{A}  \mathpunct{.}  \ottnt{B}  \mid  \lambda  \ottmv{x}  \mathbin{:}  \ottnt{A}  \mathpunct{.}  \ottnt{b}  \mid  \ottnt{b}  \gap  \ottnt{a} 
    \mid  \bot  \mid  \kw{absurd}_{ \ottnt{A} } \gap  \ottnt{b} 
    \mid  \kw{U} \gap  \ottnt{k}  \mid  \kw{Level}\texttt{<} \gap  \ell  \\
  \Gamma, \Delta & \Coloneqq  \cdot  \mid  \Gamma ,  \ottmv{x}  \mathbin{:}  \ottnt{A} 
\end{align*}
\caption{Syntax \thmref{syntactics.lean}{Term,Ctxt}}
\label{fig:syntax}
\end{figure}

The type theory is parametrized over a cofinal woset of levels,
\ie a set of levels that are well founded, totally ordered,
and each have some strictly larger level;
these properties are required when modelling the type theory.
Instances of such sets include the naturals $0, 1, 2, \dots$,
as well as the naturals extended by one limit ordinal $\omega$
and its successors $\omega + 1, \omega + 2, \dots$.
We continue to use these concrete levels for our examples.
These metalevel levels are internalized directly in system as terms $\ottmv{i}$.

\begin{figure}
\begin{mathpar}
  \inferrule[\ottdrulename{Nil}]{~}{ \mathop{\vdash}   \cdot  }
  \and
  \inferrule[\ottdrulename{Cons}]
    { \mathop{\vdash}  \Gamma  \and
      \Gamma  \vdash  \ottnt{A}  \mathrel{:}   \kw{U} \gap  \ottnt{k}  }
    { \mathop{\vdash}   \Gamma ,  \ottmv{x}  \mathbin{:}  \ottnt{A}  }
  \and
  \inferrule[\ottdrulename{Var}]
    { \mathop{\vdash}  \Gamma  \and
      \ottmv{x}  \mathrel{:}  \ottnt{A}  \in  \Gamma }
    { \Gamma  \vdash   \mathit{ \ottmv{x} }   \mathrel{:}  \ottnt{A} }
  \and
  \inferrule[\ottdrulename{Pi}]
    { \Gamma  \vdash  \ottnt{A}  \mathrel{:}   \kw{U} \gap  \ottnt{k}   \and
       \Gamma ,  \ottmv{x}  \mathbin{:}  \ottnt{A}   \vdash  \ottnt{B}  \mathrel{:}   \kw{U} \gap  \ottnt{k}  }
    { \Gamma  \vdash   \Pi  \ottmv{x}  \mathbin{:}  \ottnt{A}  \mathpunct{.}  \ottnt{B}   \mathrel{:}   \kw{U} \gap  \ottnt{k}  }
  \and
  \inferrule[\ottdrulename{Lam}]
    { \Gamma  \vdash  \ottnt{A}  \mathrel{:}   \kw{U} \gap  \ottnt{k}   \and
      \Gamma  \vdash   \Pi  \ottmv{x}  \mathbin{:}  \ottnt{A}  \mathpunct{.}  \ottnt{B}   \mathrel{:}   \kw{U} \gap  \ottnt{k}   \and
       \Gamma ,  \ottmv{x}  \mathbin{:}  \ottnt{A}   \vdash  \ottnt{b}  \mathrel{:}  \ottnt{B} }
    { \Gamma  \vdash   \lambda  \ottmv{x}  \mathbin{:}  \ottnt{A}  \mathpunct{.}  \ottnt{b}   \mathrel{:}   \Pi  \ottmv{x}  \mathbin{:}  \ottnt{A}  \mathpunct{.}  \ottnt{B}  }
  \and
  \inferrule[\ottdrulename{App}]
    { \Gamma  \vdash  \ottnt{b}  \mathrel{:}   \Pi  \ottmv{x}  \mathbin{:}  \ottnt{A}  \mathpunct{.}  \ottnt{B}   \and
      \Gamma  \vdash  \ottnt{a}  \mathrel{:}  \ottnt{A} }
    { \Gamma  \vdash   \ottnt{b}  \gap  \ottnt{a}   \mathrel{:}   \ottnt{B} [  \ottmv{x}  \mapsto  \ottnt{a}  ]  }
  \and
  \inferrule[\ottdrulename{Mty}]
    { \Gamma  \vdash   \kw{U} \gap  \ottnt{k}   \mathrel{:}   \kw{U} \gap  \ell  }
    { \Gamma  \vdash   \bot   \mathrel{:}   \kw{U} \gap  \ottnt{k}  }
  \and
  \inferrule[\ottdrulename{Abs}]
    { \Gamma  \vdash  \ottnt{A}  \mathrel{:}   \kw{U} \gap  \ottnt{k}   \and
      \Gamma  \vdash  \ottnt{b}  \mathrel{:}   \bot  }
    { \Gamma  \vdash   \kw{absurd}_{ \ottnt{A} } \gap  \ottnt{b}   \mathrel{:}  \ottnt{A} }
  \and
  \inferrule[\ottdrulename{Conv}]
    { \Gamma  \vdash  \ottnt{a}  \mathrel{:}  \ottnt{A}  \and
      \Gamma  \vdash  \ottnt{B}  \mathrel{:}   \kw{U} \gap  \ottnt{k}   \and
      \ottnt{A}  \equiv  \ottnt{B} }
    { \Gamma  \vdash  \ottnt{a}  \mathrel{:}  \ottnt{B} }
\end{mathpar}
\begin{mathpar}
  \inferrule[\ottdrulename{E-Beta}]{~}{   (  \lambda  \ottmv{x}  \mathbin{:}  \ottnt{A}  \mathpunct{.}  \ottnt{b}  )   \gap  \ottnt{a}   \equiv   \ottnt{b} [  \ottmv{x}  \mapsto  \ottnt{a}  ]  } \and
  \inferrule[\ottdrulename{E-Refl}]{~}{ \ottnt{a}  \equiv  \ottnt{a} } \and
  \inferrule[\ottdrulename{E-Sym}]{ \ottnt{a}  \equiv  \ottnt{b} }{ \ottnt{b}  \equiv  \ottnt{a} } \and
  \inferrule[\ottdrulename{E-Trans}]{ \ottnt{a}  \equiv  \ottnt{b}  \and  \ottnt{b}  \equiv  \ottnt{c} }{ \ottnt{a}  \equiv  \ottnt{c} } \and
  \cdots
\end{mathpar}
\caption{Typing and selected equality rules (no universes or levels) \thmref{typing.lean}{Wtf,Eqv}}
\label{fig:typing:basic}
\end{figure}

We begin first with the basic rules that don't concern universes or levels in \cref{fig:typing:basic},
consisting of a context well-formedness judgement \fbox{$ \mathop{\vdash}  \Gamma $},
a typing judgement \fbox{$ \Gamma  \vdash  \ottnt{a}  \mathrel{:}  \ottnt{A} $},
and an untyped definitional equality \fbox{$ \ottnt{a}  \equiv  \ottnt{b} $}.
We use $\beta$-conversion as our equality,
and omit the usual congruence rules.
Unusually, \rref{Lam} includes well-typedness premises 
of both the function's type and the domain type alone.
The former is necessary to strengthen the induction hypotheses
when proving the fundamental soundness theorem,
and the latter to strengthen them when proving subject reduction.
We later prove admissible a rule \nameref{Lam'} that omits the first premise.
The other typing rules are otherwise typical.

\begin{figure}
\begin{mathpar}
  \inferrule*[right=\ottdrulename{Univ}]
    { \Gamma  \vdash  \ottnt{k}  \mathrel{:}   \kw{Level}\texttt{<} \gap  \ell  }
    { \Gamma  \vdash   \kw{U} \gap  \ottnt{k}   \mathrel{:}   \kw{U} \gap  \ell  }
  \quad
  \inferrule*[right=\ottdrulename{Level<}]
    { \Gamma  \vdash   \kw{U} \gap  \ottnt{k_{{\mathrm{1}}}}   \mathrel{:}   \kw{U} \gap  \ell_{{\mathrm{1}}}   \and
      \Gamma  \vdash  \ottnt{k_{{\mathrm{0}}}}  \mathrel{:}   \kw{Level}\texttt{<} \gap  \ell_{{\mathrm{0}}}  }
    { \Gamma  \vdash   \kw{Level}\texttt{<} \gap  \ottnt{k_{{\mathrm{0}}}}   \mathrel{:}   \kw{U} \gap  \ottnt{k_{{\mathrm{1}}}}  }
  \quad
  \inferrule*[right=\ottdrulename{Lvl}]
    { \mathop{\vdash}  \Gamma  \and
      \ottmv{i}  <  \ottmv{j} }
    { \Gamma  \vdash   \ottmv{i}   \mathrel{:}   \kw{Level}\texttt{<} \gap   \ottmv{j}   }
  \and
  \inferrule*[right=\ottdrulename{Trans}]
    { \Gamma  \vdash  \ottnt{k_{{\mathrm{1}}}}  \mathrel{:}   \kw{Level}\texttt{<} \gap  \ottnt{k_{{\mathrm{2}}}}   \and
      \Gamma  \vdash  \ottnt{k_{{\mathrm{2}}}}  \mathrel{:}   \kw{Level}\texttt{<} \gap  \ottnt{k_{{\mathrm{3}}}}  }
    { \Gamma  \vdash  \ottnt{k_{{\mathrm{1}}}}  \mathrel{:}   \kw{Level}\texttt{<} \gap  \ottnt{k_{{\mathrm{3}}}}  }
  \and
  \inferrule*[right=\ottdrulename{Cumul}]
    { \Gamma  \vdash  \ottnt{A}  \mathrel{:}   \kw{U} \gap  \ottnt{k}   \and
      \Gamma  \vdash  \ottnt{k}  \mathrel{:}   \kw{Level}\texttt{<} \gap  \ell  }
    { \Gamma  \vdash  \ottnt{A}  \mathrel{:}   \kw{U} \gap  \ell  }
\end{mathpar}
\caption{Typing rules (universes and levels)}
\label{fig:typing:univ}
\end{figure}

The rules relating to universes and levels are given in \cref{fig:typing:univ}.
By \rref{Lvl}, we can view the type constructor $ \kw{Level}\texttt{<} \gap   \relax  $
as a restricted internalization of the order on levels.
Quantifications and abstractions over a level variable
must be bounded by some level expression,
which cannot be the variable itself since it is not in the scope of its own type.
In contrast, if we had more general level constraint types,
it would be possible to declare a looping constraint $\ottmv{x} < \ottmv{x}$.
The level type itself can be typed at any universe by \rref{Level<}
regardless of its bounding level.
For example, we can construct a derivation for $  \cdot   \vdash   \kw{Level}\texttt{<} \gap   2    \mathrel{:}   \kw{U} \gap   0   $
solely knowing that $  \cdot   \vdash   2   \mathrel{:}   \kw{Level}\texttt{<} \gap   3   $, $  \cdot   \vdash   \kw{U} \gap   0    \mathrel{:}   \kw{U} \gap   1   $,
which follow from $ 0  <  1 $ and $ 2  <  3 $.

\Rref{Trans} internalizes transitivity of the order on levels,
which is now required since levels are terms in general and not only concrete levels.
For example, we can construct a derivation for $   \ottmv{x}  \mathbin{:}   \kw{Level}\texttt{<} \gap   \omega    ,  \ottmv{y}  \mathbin{:}   \kw{Level}\texttt{<} \gap   \mathit{ \ottmv{x} }     \vdash   \mathit{ \ottmv{x} }   \mathrel{:}   \kw{Level}\texttt{<} \gap   \omega   $,
where the levels $\ottmv{x}, \ottmv{y}$ are variables.
\Rref{Cumul} is a cumulativity rule that permits lifting a type
from one universe to a higher universe.
This rule is weaker than a full subtyping rule that accounts for
contravariance in the domain and covariance in the codomain of function types.
Therefore, for instance, $  f  \mathbin{:}    \kw{U} \gap   2    \to   \kw{U} \gap   0      \vdash   \mathit{ f }   \mathrel{:}    \kw{U} \gap   1    \to   \kw{U} \gap   1    $ does \emph{not} hold.
Nonetheless, cumulativity allows us to instead type the $\eta$-expansion
$  f  \mathbin{:}    \kw{U} \gap   2    \to   \kw{U} \gap   0      \vdash    \lambda  \ottmv{x}  \mathbin{:}   \kw{U} \gap   1    \mathpunct{.}   \mathit{ f }    \gap   \mathit{ \ottmv{x} }    \mathrel{:}    \kw{U} \gap   1    \to   \kw{U} \gap   1    $.

Finally, \rref{Univ} asserts that a universe at level $\ottnt{k}$
lives in the universe at level $\ell$ when $\ottnt{k}$ is strictly bounded by $\ell$.
Allowing universes with general level terms and not just concrete levels
to be well typed is what permits typing level-polymorphic types.
For example, the level-polymorphic identity function type
$ \Pi  \ottmv{x}  \mathbin{:}   \kw{Level}\texttt{<} \gap   \omega    \mathpunct{.}    \Pi  \ottmv{y}  \mathbin{:}   \kw{U} \gap   \mathit{ \ottmv{x} }    \mathpunct{.}   \mathit{ \ottmv{y} }    \to   \mathit{ \ottmv{y} }   $ is typeable.
$ \kw{Level}\texttt{<} \gap   \omega  $ can be assigned an arbitrary type by \rref{Level},
$ \kw{U} \gap   \mathit{ \ottmv{x} }  $ has type $ \kw{U} \gap   \omega  $ by \rref{Univ} and \rref{Var},
and $\ottmv{y}$ can be assigned type $ \kw{U} \gap   \omega  $ transitively via \rref{Trans,Var}.
Then the entire term has type $ \kw{U} \gap   \omega  $ by repeated application of \rref{Pi}.

\section{Type safety} \label{sec:safety}

Type safety is proven using standard syntactic methods
to show progress and preservation (\ie subject reduction).
In essence, closed, well-typed terms evaluate (if they terminate) to values,
which are type formers and constructors,
defined below.
The proof is standard, so we omit most details,
listing only some of the key lemmas required.
\begin{equation*}
  \ottnt{v} \Coloneqq \ottmv{i} \mid  \Pi  \ottmv{x}  \mathbin{:}  \ottnt{A}  \mathpunct{.}  \ottnt{B}  \mid  \lambda  \ottmv{x}  \mathbin{:}  \ottnt{A}  \mathpunct{.}  \ottnt{b}  \mid  \bot  \mid  \kw{U} \gap  \ottnt{k}  \mid  \kw{Level}\texttt{<} \gap  \ell  \quad \text{\thmref{safety.lean}{Value}}
\end{equation*}

\subsection{Reduction and conversion}

\begin{figure}[h]
\begin{mathpar}
  \inferrule[\ottdrulename{P-Beta}]
    { \ottnt{b}  \Rightarrow  \ottnt{b'}  \and
      \ottnt{a}  \Rightarrow  \ottnt{a'} }
    {   (  \lambda  \ottmv{x}  \mathbin{:}  \ottnt{A}  \mathpunct{.}  \ottnt{b}  )   \gap  \ottnt{a}   \Rightarrow   \ottnt{b'} [  \ottmv{x}  \mapsto  \ottnt{a'}  ]  }
  \quad
  \inferrule[\ottdrulename{P-Pi}]
    { \ottnt{A}  \Rightarrow  \ottnt{A'}  \and
      \ottnt{B}  \Rightarrow  \ottnt{B'} }
    {  \Pi  \ottmv{x}  \mathbin{:}  \ottnt{A}  \mathpunct{.}  \ottnt{B}   \Rightarrow   \Pi  \ottmv{x}  \mathbin{:}  \ottnt{A'}  \mathpunct{.}  \ottnt{B'}  }
  \quad
  \inferrule[\ottdrulename{P-Lam}]
    { \ottnt{A}  \Rightarrow  \ottnt{A'}  \and
      \ottnt{b}  \Rightarrow  \ottnt{b'} }
    {  \lambda  \ottmv{x}  \mathbin{:}  \ottnt{A}  \mathpunct{.}  \ottnt{b}   \Rightarrow   \lambda  \ottmv{x}  \mathbin{:}  \ottnt{A'}  \mathpunct{.}  \ottnt{b'}  }
  \quad
  \inferrule[\ottdrulename{P-Univ}]
    { \ottnt{k}  \Rightarrow  \ottnt{k'} }
    {  \kw{U} \gap  \ottnt{k}   \Rightarrow   \kw{U} \gap  \ottnt{k'}  }
  \and
  \inferrule[\ottdrulename{P-App}]
    { \ottnt{b}  \Rightarrow  \ottnt{b'}  \quad
      \ottnt{a}  \Rightarrow  \ottnt{a'} }
    {  \ottnt{b}  \gap  \ottnt{a}   \Rightarrow   \ottnt{b'}  \gap  \ottnt{a'}  }
  \quad
  \inferrule[\ottdrulename{P-Abs}]
    { \ottnt{A}  \Rightarrow  \ottnt{A'}  \and
      \ottnt{b}  \Rightarrow  \ottnt{b'} }
    {  \kw{absurd}_{ \ottnt{A} } \gap  \ottnt{b}   \Rightarrow   \kw{absurd}_{ \ottnt{A'} } \gap  \ottnt{b'}  }
  \quad
  \inferrule[\ottdrulename{P-Level<}]
    { \ell  \Rightarrow  \ell' }
    {  \kw{Level}\texttt{<} \gap  \ell   \Rightarrow   \kw{Level}\texttt{<} \gap  \ell'  }
  \quad
  \inferrule[\ottdrulename{P-Var}]{~}{  \mathit{ \ottmv{x} }   \Rightarrow   \mathit{ \ottmv{x} }  }
  \quad
  \inferrule[\ottdrulename{P-Lvl}]{~}{  \ottmv{i}   \Rightarrow   \ottmv{i}  }
  \quad
  \inferrule[\ottdrulename{P-Mty}]{~}{  \bot   \Rightarrow   \bot  }
\end{mathpar}
\caption{Parallel reduction rules \thmref{reduction.lean}{Par,Pars}}
\label{fig:par}
\end{figure}

Rather than working directly with $\beta$-reduction,
we use parallel reduction \fbox{$ \ottnt{a}  \Rightarrow  \ottnt{b} $},
defined in \cref{fig:par},
and its reflexive, transitive closure \fbox{$ \ottnt{a}  \Rightarrow^\ast  \ottnt{b} $},
into which call-by-name evaluation embeds.
Similarly, instead of definitional equality,
we use conversion \fbox{$ \ottnt{a}  \Leftrightarrow  \ottnt{b} $},
which is defined in terms of parallel reduction.
We begin with simple lemmas about parallel reduction.

\begin{definition}[Conversion] \thmref{reduction.lean}{Conv} \\
  $ \ottnt{a}  \Leftrightarrow  \ottnt{b} $ iff there exists a $\ottnt{c}$ such that
  $ \ottnt{a}  \Rightarrow^\ast  \ottnt{c} $ and $ \ottnt{b}  \Rightarrow^\ast  \ottnt{c} $
\end{definition}

\begin{lemma}[Substitution (p.r.)] \thmref{reduction.lean}{parsSubst} \label{lem:pars:subst} \\
  If $ \ottnt{a}  \Rightarrow^\ast  \ottnt{a'} $ and $ \ottnt{b}  \Rightarrow^\ast  \ottnt{b'} $,
  then $  \ottnt{b} [  \ottmv{x}  \mapsto  \ottnt{a}  ]   \Rightarrow^\ast   \ottnt{b'} [  \ottmv{x}  \mapsto  \ottnt{a'}  ]  $.
\end{lemma}

\begin{lemma}[Construction (p.r.)] \thmref{reduction.lean}{pars\{$\beta$,Pi,Abs,$\mathcal{U}$,App,Exf,Lvl\}} \label{lem:pars:cons} \\
  Analogous constructors of parallel reduction hold
  for its reflexive, transitive closure,
  \eg if $ \ottnt{b}  \Rightarrow^\ast  \ottnt{b'} $ and $ \ottnt{a}  \Rightarrow^\ast  \ottnt{a'} $,
  then $   (  \lambda  \ottmv{x}  \mathbin{:}  \ottnt{A}  \mathpunct{.}  \ottnt{b}  )   \gap  \ottnt{a}   \Rightarrow^\ast   \ottnt{b'} [  \ottmv{x}  \mapsto  \ottnt{a'}  ]  $.
\end{lemma}

\begin{lemma}[Inversion (p.r.)] \thmref{reduction.lean}{pars\{Pi,Abs,$\mathcal{U}$,App,Exf,Lvl,Lof,Mty\}Inv} \label{lem:pars:inv} \\
  If $ \ottnt{v}  \Rightarrow^\ast  \ottnt{c} $, then $\ottnt{c}$ is also a value of the same syntactic shape
  such that the reduction is congruent,
  \eg if $  \lambda  \ottmv{x}  \mathbin{:}  \ottnt{A}  \mathpunct{.}  \ottnt{b}   \Rightarrow^\ast  \ottnt{c} $, then $\ottnt{c}$ is syntactically equal to $ \lambda  \ottmv{x}  \mathbin{:}  \ottnt{A'}  \mathpunct{.}  \ottnt{b'} $
  for some $\ottnt{A'}, \ottnt{b'}$ such that $ \ottnt{A}  \Rightarrow^\ast  \ottnt{A'} ,  \ottnt{b}  \Rightarrow^\ast  \ottnt{b'} $.
\end{lemma}

Proving that conversion is transitive requires proving confluence for parallel reduction.
We use the notion of complete development \fbox{$ \ottnt{a} ^{ \mathsf{T} } $} by Takahashi~\citep{takahashi},
which joins parallel reduction and proves the diamond property.
Its definition is omitted here,
but corresponds to simultaneous reduction of all redexes.

\begin{lemma}[Completion (p.r.)] \thmref{reduction.lean}{parTaka} \label{lem:par:compl}
  If $ \ottnt{a}  \Rightarrow  \ottnt{b} $, then $ \ottnt{b}  \Rightarrow   \ottnt{a} ^{ \mathsf{T} }  $.
\end{lemma}

\begin{corollary}[Diamond (p.r.)] \thmref{reduction.lean}{diamond} \label{lem:par:diamond}
  If $ \ottnt{a}  \Rightarrow  \ottnt{b} $ and $ \ottnt{a}  \Rightarrow  \ottnt{c} $,
  then there exists some $d$ such that $ \ottnt{b}  \Rightarrow   \mathit{ d }  $ and $ \ottnt{c}  \Rightarrow   \mathit{ d }  $.
  In particular, $d$ is $ \ottnt{a} ^{ \mathsf{T} } $,
  with the reductions given by \nameref{lem:par:compl}.
\end{corollary}

\begin{theorem}[Confluence (p.r.)] \thmref{reduction.lean}{confluence} \label{lem:par:confl} \\
  If $ \ottnt{a}  \Rightarrow^\ast  \ottnt{b} $ and $ \ottnt{a}  \Rightarrow^\ast  \ottnt{c} $,
  then there exists some $d$ such that $ \ottnt{b}  \Rightarrow^\ast   \mathit{ d }  $ and $ \ottnt{c}  \Rightarrow^\ast   \mathit{ d }  $.
\end{theorem}

\begin{corollary}[Properties of conversion] \thmref{reduction.lean}{conv*} \label{lem:conv}
  Conversion is reflexive, symmetric, transitive, substitutive, and congruent.
  Transitivity requires \nameref{lem:par:confl};
  the remaining are straightforward
  from the corresponding properties of parallel reduction.
\end{corollary}

Inversion on parallel reduction gives syntactic consistency and injectivity of conversion.
Finally, definitional equality is equivalent to conversion,
which allows us to use them interchangeably later on.

\begin{lemma}[Syntactic consistency] \thmref{reduction.lean}{conv\{$\mathcal{U}$,Pi,Mty,Lvl\}\{$\mathcal{U}$,Pi,Mty,Lvl\}} \label{lem:par:consistency}
  If $\ottnt{v_{{\mathrm{1}}}}$ and $\ottnt{v_{{\mathrm{2}}}}$ have different syntactic shapes,
  then $ \ottnt{v_{{\mathrm{1}}}}  \Leftrightarrow  \ottnt{v_{{\mathrm{2}}}} $ is impossible.
\end{lemma}

\begin{lemma}[Injectivity (conv.)] \thmref{reduction.lean}{conv\{Pi,$\mathcal{U}$,Lvl\}Inv}
  \begin{enumerate}[topsep=0pt]
    \item If $  \Pi  \ottmv{x}  \mathbin{:}  \ottnt{A_{{\mathrm{1}}}}  \mathpunct{.}  \ottnt{B_{{\mathrm{1}}}}   \Leftrightarrow   \Pi  \ottmv{x}  \mathbin{:}  \ottnt{A_{{\mathrm{2}}}}  \mathpunct{.}  \ottnt{B_{{\mathrm{2}}}}  $, then $ \ottnt{A_{{\mathrm{1}}}}  \Leftrightarrow  \ottnt{A_{{\mathrm{2}}}} $ and $ \ottnt{B_{{\mathrm{1}}}}  \Leftrightarrow  \ottnt{B_{{\mathrm{2}}}} $.
    \item If $  \kw{U} \gap  \ottnt{k_{{\mathrm{1}}}}   \Leftrightarrow   \kw{U} \gap  \ottnt{k_{{\mathrm{2}}}}  $, then $ \ottnt{k_{{\mathrm{1}}}}  \Leftrightarrow  \ottnt{k_{{\mathrm{2}}}} $.
    \item If $  \kw{Level}\texttt{<} \gap  \ottnt{k_{{\mathrm{1}}}}   \Leftrightarrow   \kw{Level}\texttt{<} \gap  \ottnt{k_{{\mathrm{2}}}}  $, then $ \ottnt{k_{{\mathrm{1}}}}  \Leftrightarrow  \ottnt{k_{{\mathrm{2}}}} $.
  \end{enumerate}
\end{lemma}

\begin{theorem} \thmref{typing.lean}{convEqv,eqvConv} \label{lem:eq-conv}
  $ \ottnt{a}  \equiv  \ottnt{b} $ iff $ \ottnt{a}  \Leftrightarrow  \ottnt{b} $.
\end{theorem}

\subsection{Subject reduction and type safety}

To prove subject reduction,
we need the usual weakening, substitution, replacement, and regularity lemmas.
They follow from stronger forms of these lemmas involving simultaneous renaming and substitution,
whose details we omit.

\begin{lemma} \thmref{safety.lean}{wtWeaken,wtSubst,wtReplace,wtRegularity} \label{lem:wt:preservation}
  \begin{itemize}
    \item \textit{Weakening.} If $ \mathop{\vdash}  \Gamma $, $ \Gamma  \vdash  \ottnt{B}  \mathrel{:}   \kw{U} \gap  \ottnt{k}  $, and $ \Gamma  \vdash  \ottnt{a}  \mathrel{:}  \ottnt{A} $,
      then $  \Gamma ,  \ottmv{x}  \mathbin{:}  \ottnt{B}   \vdash  \ottnt{a}  \mathrel{:}  \ottnt{A} $, where $\ottmv{x}$ not in $\ottnt{a}$, $\ottnt{A}$.
    \item \textit{Substitution.} If $ \Gamma  \vdash  \ottnt{b}  \mathrel{:}  \ottnt{B} $ and $  \Gamma ,  \ottmv{x}  \mathbin{:}  \ottnt{B}   \vdash  \ottnt{a}  \mathrel{:}  \ottnt{A} $,
      then $ \Gamma  \vdash   \ottnt{a} [  \ottmv{x}  \mapsto  \ottnt{b}  ]   \mathrel{:}   \ottnt{A} [  \ottmv{x}  \mapsto  \ottnt{b}  ]  $.
    \item \textit{Replacement.} If $ \ottnt{A}  \equiv  \ottnt{B} $, $ \Gamma  \vdash  \ottnt{B}  \mathrel{:}   \kw{U} \gap  \ottnt{k}  $, and $  \Gamma ,  \ottmv{x}  \mathbin{:}  \ottnt{A}   \vdash  \ottnt{c}  \mathrel{:}  \ottnt{C} $,
      then $  \Gamma ,  \ottmv{x}  \mathbin{:}  \ottnt{B}   \vdash  \ottnt{c}  \mathrel{:}  \ottnt{C} $.
    \item \textit{Regularity.} If $ \Gamma  \vdash  \ottnt{a}  \mathrel{:}  \ottnt{A} $, then there exists some $\ottnt{k}$ such that
      $ \Gamma  \vdash  \ottnt{A}  \mathrel{:}   \kw{U} \gap  \ottnt{k}  $.
  \end{itemize}
\end{lemma}

\begin{theorem}[Subject reduction] \thmref{safety.lean}{wtPar} \label{lem:preservation} \\
  If $ \ottnt{a}  \Rightarrow  \ottnt{b} $ and $ \Gamma  \vdash  \ottnt{a}  \mathrel{:}  \ottnt{A} $, then $ \Gamma  \vdash  \ottnt{b}  \mathrel{:}  \ottnt{A} $.
\end{theorem}

\begin{proof}
  By induction on the typing derivation of $\ottnt{a}$.
  The most complex case is when the reduction is \rref*{P-Beta},
  requiring \cref{lem:conv} and \cref{lem:wt:preservation}.
  Even so, the proof is standard,
  and the cases for the universe and level rules in \cref{fig:typing:univ}
  follow from the induction hypotheses.
\end{proof}

At this point, we are able to prove admissibility of \rref{Lam}
without its first premise, which depends only on regularity.

\begin{corollary}[\textsc{Lam'}] \thmref{safety.lean}{wtfAbs} \label{Lam'} \\
  Given $ \Gamma  \vdash   \Pi  \ottmv{x}  \mathbin{:}  \ottnt{A}  \mathpunct{.}  \ottnt{B}   \mathrel{:}   \kw{U} \gap  \ottnt{k}  $ and $  \Gamma ,  \ottmv{x}  \mathbin{:}  \ottnt{A}   \vdash  \ottnt{b}  \mathrel{:}  \ottnt{B} $,
  we have $ \Gamma  \vdash   \lambda  \ottmv{x}  \mathbin{:}  \ottnt{A}  \mathpunct{.}  \ottnt{b}   \mathrel{:}   \Pi  \ottmv{x}  \mathbin{:}  \ottnt{A}  \mathpunct{.}  \ottnt{B}  $.
\end{corollary}

For progress and type safety, our notion of evaluation is
the reflexive, transitive closure \fbox{$ \ottnt{a}  \rightsquigarrow^\ast  \ottnt{b} $}
of call-by-name (cbn) reduction \fbox{$ \ottnt{a}  \rightsquigarrow  \ottnt{b} $},
which reduces $\beta$-redexes and head positions.
A single step of cbn reduction embeds into
a single step of parallel reduction by induction,
which allows us to use \nameref{lem:preservation}.
These proofs are also standard.

\begin{lemma}[Progress] \thmref{safety.lean}{wtProgress} \label{lem:progress} \\
  If $  \cdot   \vdash  \ottnt{a}  \mathrel{:}  \ottnt{A} $, then either $\ottnt{a}$ is a value,
  or $ \ottnt{a}  \rightsquigarrow  \ottnt{b} $ for some $\ottnt{b}$.
\end{lemma}

\begin{theorem}[Type safety] \thmref{safety.lean}{wtSafety} \\
  If $  \cdot   \vdash  \ottnt{a}  \mathrel{:}  \ottnt{A} $ and $ \ottnt{a}  \rightsquigarrow^\ast  \ottnt{b} $,
  then either $\ottnt{b}$ is a value,
  or $ \ottnt{b}  \rightsquigarrow  \ottnt{c} $ for some $\ottnt{c}$.
\end{theorem}

\section{Consistency and canonicity} \label{sec:lr}

To prove consistency and canonicity,
we use a logical relation to semantically interpret closed types as sets of closed terms;
these sets are backward closed under reduction,
so if a term reduces to something in the set, then it is also in the set.
The empty type is interpreted as the empty set,
universes as sets of terms that reduce to types,
and level types as sets of terms that reduce to concrete levels.
Consistency and canonicity then follow from the fundamental soundness theorem,
which states that if a term $\ottnt{a}$ has type $\ottnt{A}$,
then $\ottnt{a}$ is in the interpretation of $\ottnt{A}$.
For instance, there is no closed term of the empty type,
since it must belong to its interpretation as an empty set, which is a contradiction.
The structure of the logical relation and the soundness proof
is adapted from the mechanization by Liu~\citep{lr-pearl}.
We cover some details here,
especially as they pertain to universes and levels.

\subsection{Logical relation for closed types}

The logical relation is written as \fbox{$ \mathopen{\llbracket}  \ottnt{A}  \mathclose{\rrbracket}_{ \ottmv{i} } \searrow  \ottnt{P} $},
where $\ottnt{A}$ is the type, $\ottnt{P}$ is the set of terms,
and $\ottmv{i}$ is the universe level of the type.
A set of terms $\ottnt{P}$ is mechanized as a predicate on terms,
though we to write $ \ottnt{a}  \in  \ottnt{P} $ in lieu of $\ottnt{P}(\ottnt{a})$
to say that $\ottnt{a}$ is in the set,
and we use set-builder notation in lieu of explicit abstractions.
When proving properties of the logical relation,
we require no other axioms than predicate extensionality,
which follows from function and propositional extensionality;
we explicitly mark the lemmas in which they are used with $\dagger$.

Because universes are interpreted as sets of types
which themselves have interpretations at a lower universe level,
to ensure that the interpretation is well defined,
the mechanization implements it as an inductive definition
parametrized by interpretations at lower levels,
then ties the knot by well-founded induction on levels.
For clarity and concision, we ignore these details
and present the logical relation in \cref{fig:lr:closed}
without worrying about well-foundedness.

\begin{figure}
\begin{mathpar}
  \inferrule[\ottdrulename{I-Mty}]{~}
    { \mathopen{\llbracket}   \bot   \mathclose{\rrbracket}_{ \ottmv{i} } \searrow   \varnothing  }
  \and
  \inferrule[\ottdrulename{I-Univ}]
    { \ottmv{j}  <  \ottmv{i} }
  { \mathopen{\llbracket}   \kw{U} \gap   \ottmv{j}    \mathclose{\rrbracket}_{ \ottmv{i} } \searrow   \lbrace  \ottmv{z}  \mid   \exists  \ottnt{P}  \mathpunct{.}   \mathopen{\llbracket}   \mathit{ \ottmv{z} }   \mathclose{\rrbracket}_{ \ottmv{j} } \searrow  \ottnt{P}    \rbrace  }
  \and
  \inferrule[\ottdrulename{I-Level<}]{~}
    { \mathopen{\llbracket}   \kw{Level}\texttt{<} \gap   \ottmv{j_{{\mathrm{1}}}}    \mathclose{\rrbracket}_{ \ottmv{i} } \searrow   \lbrace  \ottmv{z}  \mid   \exists  \ottmv{j_{{\mathrm{2}}}}  \mathpunct{.}     \mathit{ \ottmv{z} }   \Rightarrow^\ast   \ottmv{j_{{\mathrm{2}}}}    \wedge   \ottmv{j_{{\mathrm{2}}}}  <  \ottmv{j_{{\mathrm{1}}}}     \rbrace  }
  \and
  \inferrule[\ottdrulename{I-Step}]
    { \ottnt{A}  \Rightarrow  \ottnt{B}  \and
      \mathopen{\llbracket}  \ottnt{B}  \mathclose{\rrbracket}_{ \ottmv{i} } \searrow  \ottnt{P} }
    { \mathopen{\llbracket}  \ottnt{A}  \mathclose{\rrbracket}_{ \ottmv{i} } \searrow  \ottnt{P} }
  \and
  \inferrule[\ottdrulename{I-Pi}]
    { \mathopen{\llbracket}  \ottnt{A}  \mathclose{\rrbracket}_{ \ottmv{i} } \searrow  \ottnt{P_{{\mathrm{1}}}}  \and
        \forall  \ottmv{y}  \mathpunct{.}   \mathit{ \ottmv{y} }    \in  \ottnt{P_{{\mathrm{1}}}}   \to   \exists  \ottnt{P_{{\mathrm{2}}}}  \mathpunct{.}   \ottnt{R} ( \ottmv{y} ,  \ottnt{P_{{\mathrm{2}}}} )    \\\\
       \forall  \ottmv{y}  \mathpunct{.}   \forall  \ottnt{P_{{\mathrm{2}}}}  \mathpunct{.}   \ottnt{R} ( \ottmv{y} ,  \ottnt{P_{{\mathrm{2}}}} )     \to   \mathopen{\llbracket}   \ottnt{B} [  \ottmv{x}  \mapsto   \mathit{ \ottmv{y} }   ]   \mathclose{\rrbracket}_{ \ottmv{i} } \searrow  \ottnt{P_{{\mathrm{2}}}}  }
    { \mathopen{\llbracket}   \Pi  \ottmv{x}  \mathbin{:}  \ottnt{A}  \mathpunct{.}  \ottnt{B}   \mathclose{\rrbracket}_{ \ottmv{i} } \searrow   \lbrace  f  \mid    \forall  \ottmv{y}  \mathpunct{.}   \forall  \ottnt{P_{{\mathrm{2}}}}  \mathpunct{.}   \ottnt{R} ( \ottmv{y} ,  \ottnt{P_{{\mathrm{2}}}} )     \to     \mathit{ \ottmv{y} }   \in  \ottnt{P_{{\mathrm{1}}}}   \to     \mathit{ f }   \gap   \mathit{ \ottmv{y} }    \in  \ottnt{P_{{\mathrm{2}}}}     \rbrace  }
\end{mathpar}
\caption{Logical relation for closed types \thmref{semantics.lean}{Interps}}
\label{fig:lr:closed}
\end{figure}

Let us get the easier cases out of the way.
The interpretation of the empty type as the empty set is given by \rref{I-Mty}.
\Rref{I-Step} backward closes the interpretation under reduction of the type,
so a type has an interpretation if it reduces to a type with an interpretation.
We show shortly that forward closure under reduction of the type also holds,
as well as backward closure under reduction of the \emph{terms} in the interpretations.%
\footnote{We do not require forward closure.}

Because we consider the interpretation of closed types only,
and we have a constructor for backward closure,
the only other constructors we need are those for normal, closed types.
In particular, we need only consider $ \kw{U} \gap   \ottmv{j}  $ and $ \kw{Level}\texttt{<} \gap   \ottmv{j_{{\mathrm{1}}}}  $
with concrete levels rather than arbitrary level terms.
The interpretation of $ \kw{Level}\texttt{<} \gap   \ottmv{j_{{\mathrm{1}}}}  $ given by \rref{I-Level<}
is the set of level terms strictly less than $\ottmv{j_{{\mathrm{1}}}}$;
more precisely, it is the set of terms that reduce to such concrete levels.
The interpretation of $ \kw{U} \gap   \ottmv{j}  $ given by \rref{I-Univ}
is the set of types that have an interpretation.

The intuition behind \rref{I-Pi} for function types is that a function $f$
is in its interpretation if for every argument $\ottmv{y}$ in the interpretation of the domain,
the application $  \mathit{ f }   \gap   \mathit{ \ottmv{y} }  $ is in the interpretation of the codomain.
Because we are dealing with dependent types,
the interpretation of the codomain varies with the argument,
so we need to ensure first that the interpretation exists
for \emph{every} argument in the interpretation of the domain,
and that $  \mathit{ f }   \gap   \mathit{ \ottmv{y} }  $ is in the \emph{particular} interpretation of the codomain.
It then sounds like we would want \rref{I-Pi'} below \thmref{semantics.lean}{interpsPi}.
\begin{mathpar}
  \inferrule*[right=\ottdrulename{I-Pi'}]
    { \mathopen{\llbracket}  \ottnt{A}  \mathclose{\rrbracket}_{ \ottmv{i} } \searrow  \ottnt{P_{{\mathrm{1}}}}  \and
        \forall  \ottmv{y}  \mathpunct{.}   \mathit{ \ottmv{y} }    \in  \ottnt{P_{{\mathrm{1}}}}   \to   \exists  \ottnt{P_{{\mathrm{2}}}}  \mathpunct{.}   \mathopen{\llbracket}   \ottnt{B} [  \ottmv{x}  \mapsto   \mathit{ \ottmv{y} }   ]   \mathclose{\rrbracket}_{ \ottmv{i} } \searrow  \ottnt{P_{{\mathrm{2}}}}   }
    { \mathopen{\llbracket}   \Pi  \ottmv{x}  \mathbin{:}  \ottnt{A}  \mathpunct{.}  \ottnt{B}   \mathclose{\rrbracket}_{ \ottmv{i} } \searrow   \lbrace  f  \mid   \forall  \ottmv{y}  \mathpunct{.}     \forall  \ottnt{P_{{\mathrm{2}}}}  \mathpunct{.}   (  \mathopen{\llbracket}   \ottnt{B} [  \ottmv{x}  \mapsto   \mathit{ \ottmv{y} }   ]   \mathclose{\rrbracket}_{ \ottmv{i} } \searrow  \ottnt{P_{{\mathrm{2}}}}  )    \to    \mathit{ \ottmv{y} }   \in  \ottnt{P_{{\mathrm{1}}}}    \to     \mathit{ f }   \gap   \mathit{ \ottmv{y} }    \in  \ottnt{P_{{\mathrm{2}}}}     \rbrace  }
\end{mathpar}

The problem is that the interpretation is not strictly positive in the conclusion,
so \rref*{I-Pi'} as a constructor is not well defined.
\Rref{I-Pi} therefore uses an auxiliary relation $\ottnt{R}$
that relates the argument $\ottmv{y}$ to the interpretation of the codomain $ \ottnt{B} [  \ottmv{x}  \mapsto   \mathit{ \ottmv{y} }   ] $.
\Rref{I-Pi'} then holds by instantiating $ \ottnt{R} ( \ottmv{y} ,  \ottnt{P_{{\mathrm{2}}}} ) $ with $ \mathopen{\llbracket}   \ottnt{B} [  \ottmv{x}  \mapsto   \mathit{ \ottmv{y} }   ]   \mathclose{\rrbracket}_{ \ottmv{i} } \searrow  \ottnt{P_{{\mathrm{2}}}} $ in \rref{I-Pi}.
This is the same trick used by Liu~\citep{lr-pearl},
whose origins are documented by Anand and Rahli~\citep{mech-nuprl}.

We require of the logical relation inversion properties for each constructor,
along with properties that hold \apriori for syntactic typing:
conversion and cumulativity.
A key intermediate lemma is functionality,
\ie that the interpretation of a type is deterministic.
Cumulativity holds directly by induction on the logical relation.
To prove conversion, we begin with closures over reductions.

\begin{lemma}[Cumulativity (l.r.)] \thmref{semantics.lean}{interpsCumul} \label{lem:lr:cumul} \\
  Suppose $ \ottmv{i}  <  \ottmv{j} $. If $ \mathopen{\llbracket}  \ottnt{A}  \mathclose{\rrbracket}_{ \ottmv{i} } \searrow  \ottnt{P} $, then $ \mathopen{\llbracket}  \ottnt{A}  \mathclose{\rrbracket}_{ \ottmv{j} } \searrow  \ottnt{P} $.
\end{lemma}

\begin{lemma}[Forward and backward closure (l.r.)] \thmref{semantics.lean}{interps\{Fwds,Bwds\}} \label{lem:lr:pars} ~
  \begin{enumerate}[topsep=0pt]
    \item If $ \mathopen{\llbracket}  \ottnt{A}  \mathclose{\rrbracket}_{ \ottmv{i} } \searrow  \ottnt{P} $ and either $ \ottnt{A}  \Rightarrow  \ottnt{B} $ or $ \ottnt{A}  \Rightarrow^\ast  \ottnt{B} $,
      then $ \mathopen{\llbracket}  \ottnt{B}  \mathclose{\rrbracket}_{ \ottmv{i} } \searrow  \ottnt{P} $.
    \item If $ \mathopen{\llbracket}  \ottnt{B}  \mathclose{\rrbracket}_{ \ottmv{i} } \searrow  \ottnt{P} $ and either $ \ottnt{A}  \Rightarrow  \ottnt{B} $ or $ \ottnt{A}  \Rightarrow^\ast  \ottnt{B} $,
      then $ \mathopen{\llbracket}  \ottnt{A}  \mathclose{\rrbracket}_{ \ottmv{i} } \searrow  \ottnt{P} $.
  \end{enumerate}
\end{lemma}

\begin{proof} ~
  \begin{enumerate}[topsep=0pt]
    \item For $ \ottnt{A}  \Rightarrow  \ottnt{B} $, by induction on the logical relation,
      using \nameref{lem:par:diamond} in the \rref*{I-Step} case.
      \nameref{lem:pars:subst} is needed in the \rref*{I-Pi} case
      to manipulate the substitution in the function codomain.
      For $ \ottnt{A}  \Rightarrow^\ast  \ottnt{B} $, by induction on this reduction.
    \item For $ \ottnt{A}  \Rightarrow  \ottnt{B} $, directly by \rref{I-Step}.
      For $ \ottnt{A}  \Rightarrow^\ast  \ottnt{B} $, by induction on this reduction. \qedhere
  \end{enumerate}
\end{proof}

\begin{corollary}[Conversion (l.r.)] \thmref{semantics.lean}{interpsConv} \\
  If $ \mathopen{\llbracket}  \ottnt{A}  \mathclose{\rrbracket}_{ \ottmv{i} } \searrow  \ottnt{P} $ and $ \ottnt{A}  \Leftrightarrow  \ottnt{B} $,
  then $ \mathopen{\llbracket}  \ottnt{B}  \mathclose{\rrbracket}_{ \ottmv{i} } \searrow  \ottnt{P} $,
  using forward and backward closure.
\end{corollary}

The final closure lemma we need is backward closure of the terms in the interpretations.
When proving the fundamental theorem,
we encounter situations where our goal requires inclusion of a reduced term in an interpretation,
while induction hypotheses only piece together inclusion of the term before reduction.

\begin{lemma}[Backward closure] \thmref{semantics.lean}{interpsBwdsP} \label{lem:lr:back} \\
  If $ \mathopen{\llbracket}  \ottnt{A}  \mathclose{\rrbracket}_{ \ottmv{i} } \searrow  \ottnt{P} $ and $ \ottnt{a}  \Rightarrow^\ast  \ottnt{b} $,
  then $ \ottnt{b}  \in  \ottnt{P} $ implies $ \ottnt{a}  \in  \ottnt{P} $.
\end{lemma}

\begin{proof}
  By induction on the logical relation.
  In the \rref*{I-Univ} case, where $\ottnt{a}$ and $\ottnt{b}$ are types,
  we use backward closure from \cref{lem:lr:pars}.
\end{proof}

The inversion principles for each constructor of the logical relation
hold by induction, using properties of parallel reduction as needed.
However, it is the inversion principle for \rref{I-Pi'} that we want.
The issue lies in the set of terms of the interpretation:
if we do not yet know that the sets are unique,
then inversion on \rref{I-Pi} gives \emph{some} interpretation $\ottnt{P_{{\mathrm{2}}}}$ of the codomain,
but we do not know whether it is \emph{the} interpretation that is required.
We solve this by proving functionality.

\begin{lemma}[Fixed-level functionality (l.r.)]$\!\!{\dagger}$ \thmref{semantics.lean}{interpsDet'} \label{lem:lr:fixed-func} \\
  If $ \mathopen{\llbracket}  \ottnt{A}  \mathclose{\rrbracket}_{ \ottmv{i} } \searrow  \ottnt{P} $ and $ \mathopen{\llbracket}  \ottnt{A}  \mathclose{\rrbracket}_{ \ottmv{i} } \searrow  \ottnt{Q} $, then $\ottnt{P} = \ottnt{Q}$.
\end{lemma}

\begin{proof}
  By induction on the first logical relation,
  then generally inversion on the second,
  except for the \rref*{I-Step} case,
  which holds directly by the induction hypothesis
  and forward closure on the second logical relation.
  The complex case is \rref*{I-Pi},
  where we must prove the two sets of terms equal,
  knowing by the induction hypotheses
  that the interpretations of the domain and codomain yield equal sets.
  Because sets are encoded as predicates,
  we need to use predicate extensionality.
  It then suffices to show that membership in one set implies membership in the other,
  which holds using the induction hypotheses.
\end{proof}

Functionality holds even with different universe levels,
the idea being that the interpretation of a type is independent
of the level at which it lives.
We are then finally able to prove the inversion property for \rref{I-Pi'}.

\begin{lemma}[Functionality (l.r.)] \thmref{semantics.lean}{interpsDet} \label{lem:lr:func} \\
  If $ \mathopen{\llbracket}  \ottnt{A}  \mathclose{\rrbracket}_{ \ottmv{i} } \searrow  \ottnt{P} $ and $ \mathopen{\llbracket}  \ottnt{A}  \mathclose{\rrbracket}_{ \ottmv{j} } \searrow  \ottnt{Q} $, then $\ottnt{P} = \ottnt{Q}$.
\end{lemma}

\begin{proof}
  By totality of the order on levels,
  either $\ottmv{i}$ and $\ottmv{j}$ are equal,
  or one is strictly larger than the other.
  In the latter case,
  we use \nameref{lem:lr:cumul} to lift the logical relation at the lower level to the higher level.
  Then the sets are equal by \nameref{lem:lr:fixed-func}.
\end{proof}

\begin{lemma}[Inversion on function types (l.r.)]$\!\!{\dagger}$ \thmref{semantics.lean}{interpsPiInv} \label{lem:lr:inv-pi} \\
  If $ \mathopen{\llbracket}   \Pi  \ottmv{x}  \mathbin{:}  \ottnt{A}  \mathpunct{.}  \ottnt{B}   \mathclose{\rrbracket}_{ \ottmv{i} } \searrow  \ottnt{P} $,
  then there exists a $\ottnt{P_{{\mathrm{1}}}}$ such that:
  \begin{enumerate}[topsep=0pt]
    \item \label{lem:inv-pi:goal:A} $ \mathopen{\llbracket}  \ottnt{A}  \mathclose{\rrbracket}_{ \ottmv{i} } \searrow  \ottnt{P_{{\mathrm{1}}}} $;
    \item \label{lem:inv-pi:goal:B} $   \forall  \ottmv{y}  \mathpunct{.}   \mathit{ \ottmv{y} }    \in  \ottnt{P_{{\mathrm{1}}}}   \to   \exists  \ottnt{P_{{\mathrm{2}}}}  \mathpunct{.}   \mathopen{\llbracket}   \ottnt{B} [  \ottmv{x}  \mapsto   \mathit{ \ottmv{y} }   ]   \mathclose{\rrbracket}_{ \ottmv{i} } \searrow  \ottnt{P_{{\mathrm{2}}}}   $; and
    \item \label{lem:inv-pi:goal:P} $\ottnt{P} =  \lbrace  f  \mid   \forall  \ottmv{y}  \mathpunct{.}     \forall  \ottnt{P_{{\mathrm{2}}}}  \mathpunct{.}   (  \mathopen{\llbracket}   \ottnt{B} [  \ottmv{x}  \mapsto   \mathit{ \ottmv{y} }   ]   \mathclose{\rrbracket}_{ \ottmv{i} } \searrow  \ottnt{P_{{\mathrm{2}}}}  )    \to    \mathit{ \ottmv{y} }   \in  \ottnt{P_{{\mathrm{1}}}}    \to     \mathit{ f }   \gap   \mathit{ \ottmv{y} }    \in  \ottnt{P_{{\mathrm{2}}}}     \rbrace $.
  \end{enumerate}
\end{lemma}

\begin{proof}
  By inversion on the logical relation,
  which gives $\ottnt{P_{{\mathrm{1}}}}$ and $\ottnt{R}$ such that:
  \begin{enumerate}[topsep=0pt,start=4]
    \item \label{lem:inv-pi:hyp:A} $ \mathopen{\llbracket}  \ottnt{A}  \mathclose{\rrbracket}_{ \ottmv{i} } \searrow  \ottnt{P_{{\mathrm{1}}}} $;
    \item \label{lem:inv-pi:hyp:R} $   \forall  \ottmv{y}  \mathpunct{.}   \mathit{ \ottmv{y} }    \in  \ottnt{P_{{\mathrm{1}}}}   \to   \exists  \ottnt{P_{{\mathrm{2}}}}  \mathpunct{.}   \ottnt{R} ( \ottmv{y} ,  \ottnt{P_{{\mathrm{2}}}} )   $;
    \item \label{lem:inv-pi:hyp:B} $  \forall  \ottmv{y}  \mathpunct{.}   \forall  \ottnt{P_{{\mathrm{2}}}}  \mathpunct{.}   \ottnt{R} ( \ottmv{y} ,  \ottnt{P_{{\mathrm{2}}}} )     \to   \mathopen{\llbracket}   \ottnt{B} [  \ottmv{x}  \mapsto   \mathit{ \ottmv{y} }   ]   \mathclose{\rrbracket}_{ \ottmv{i} } \searrow  \ottnt{P_{{\mathrm{2}}}}  $; and
    \item \label{lem:inv-pi:hyp:P} $\ottnt{P} =  \lbrace  f  \mid    \forall  \ottmv{y}  \mathpunct{.}   \forall  \ottnt{P_{{\mathrm{2}}}}  \mathpunct{.}   \ottnt{R} ( \ottmv{y} ,  \ottnt{P_{{\mathrm{2}}}} )     \to     \mathit{ \ottmv{y} }   \in  \ottnt{P_{{\mathrm{1}}}}   \to     \mathit{ f }   \gap   \mathit{ \ottmv{y} }    \in  \ottnt{P_{{\mathrm{2}}}}     \rbrace $.
  \end{enumerate}
  \ref{lem:inv-pi:goal:A} holds directly by \ref{lem:inv-pi:hyp:A},
  and \ref{lem:inv-pi:goal:B} holds by combining \ref{lem:inv-pi:hyp:R} and \ref{lem:inv-pi:hyp:B}.
  To show that the sets in \ref{lem:inv-pi:goal:P} and \ref{lem:inv-pi:hyp:P} are equal,
  we again use predicate extensionality.
  \begin{itemize}[topsep=0pt]
    \item \textit{\ref{lem:inv-pi:goal:P} implies \ref{lem:inv-pi:hyp:P}.}
      Supposing $\ottmv{y}$ and $\ottnt{P_{{\mathrm{2}}}}$,
      we have three hypotheses $   (  \mathopen{\llbracket}   \ottnt{B} [  \ottmv{x}  \mapsto   \mathit{ \ottmv{y} }   ]   \mathclose{\rrbracket}_{ \ottmv{i} } \searrow  \ottnt{P_{{\mathrm{2}}}}  )   \to    \mathit{ \ottmv{y} }   \in  \ottnt{P_{{\mathrm{1}}}}    \to     \mathit{ f }   \gap   \mathit{ \ottmv{y} }    \in  \ottnt{P_{{\mathrm{2}}}}  $,
      $ \ottnt{R} ( \ottmv{y} ,  \ottnt{P_{{\mathrm{2}}}} ) $, and $  \mathit{ \ottmv{y} }   \in  \ottnt{P_{{\mathrm{1}}}} $.
      From \ref{lem:inv-pi:hyp:B} on the second hypothesis,
      we have $ \mathopen{\llbracket}   \ottnt{B} [  \ottmv{x}  \mapsto   \mathit{ \ottmv{y} }   ]   \mathclose{\rrbracket}_{ \ottmv{i} } \searrow  \ottnt{P_{{\mathrm{2}}}} $,
      so we can apply the first hypothesis to get $   \mathit{ f }   \gap   \mathit{ \ottmv{y} }    \in  \ottnt{P_{{\mathrm{2}}}} $.
    \item \textit{\ref{lem:inv-pi:hyp:P} implies \ref{lem:inv-pi:goal:P}.}
      Supposing $\ottmv{y}$ and $\ottnt{P_{{\mathrm{2}}}}$,
      we have three hypotheses $  \ottnt{R} ( \ottmv{y} ,  \ottnt{P_{{\mathrm{2}}}} )   \to     \mathit{ \ottmv{y} }   \in  \ottnt{P_{{\mathrm{1}}}}   \to     \mathit{ f }   \gap   \mathit{ \ottmv{y} }    \in  \ottnt{P_{{\mathrm{2}}}}   $,
      $ \mathopen{\llbracket}   \ottnt{B} [  \ottmv{x}  \mapsto   \mathit{ \ottmv{y} }   ]   \mathclose{\rrbracket}_{ \ottmv{i} } \searrow  \ottnt{P_{{\mathrm{2}}}} $, and $  \mathit{ \ottmv{y} }   \in  \ottnt{P_{{\mathrm{1}}}} $.
      By the first hypothesis on the second and on \ref{lem:inv-pi:hyp:R},
      there exists a $\ottnt{P'_{{\mathrm{2}}}}$ such that $   \mathit{ f }   \gap   \mathit{ \ottmv{y} }    \in  \ottnt{P'_{{\mathrm{2}}}} $.
      From \ref{lem:inv-pi:hyp:B}, we also have $ \mathopen{\llbracket}   \ottnt{B} [  \ottmv{x}  \mapsto   \mathit{ \ottmv{y} }   ]   \mathclose{\rrbracket}_{ \ottmv{i} } \searrow  \ottnt{P'_{{\mathrm{2}}}} $.
      Then by \nameref{lem:lr:func}, we have $\ottnt{P_{{\mathrm{2}}}} = \ottnt{P'_{{\mathrm{2}}}}$, so $   \mathit{ f }   \gap   \mathit{ \ottmv{y} }    \in  \ottnt{P_{{\mathrm{2}}}} $.
      \qedhere
  \end{itemize}
\end{proof}

Inversion principles also hold for the other types
by induction on the logical relation.

\begin{lemma}[Inversion on universes (l.r.)] \thmref{semantics.lean}{interps$\mathcal{U}$Inv} \label{lem:lr:inv-univ} \\
  If $ \mathopen{\llbracket}   \kw{U} \gap  \ottnt{k}   \mathclose{\rrbracket}_{ \ottmv{i} } \searrow  \ottnt{P} $ and $ \ottnt{A}  \in  \ottnt{P} $,
  then there exists $\ottmv{j}, \ottnt{Q}$ such that $ \ottnt{k}  \Rightarrow^\ast   \ottmv{j}  $ and $ \mathopen{\llbracket}  \ottnt{A}  \mathclose{\rrbracket}_{ \ottmv{j} } \searrow  \ottnt{Q} $.
\end{lemma}

\begin{lemma}[Inversion on level types (l.r.)] \thmref{semantics.lean}{interpLvlInv} \label{lem:lr:inv-lvl} \\
  If $ \mathopen{\llbracket}   \kw{Level}\texttt{<} \gap  \ell   \mathclose{\rrbracket}_{ \ottmv{i} } \searrow  \ottnt{P} $ and $ \ottnt{k}  \in  \ottnt{P} $,
  then there exist $\ottmv{j_{{\mathrm{2}}}} < \ottmv{j_{{\mathrm{1}}}}$ such that $ \ell  \Rightarrow^\ast   \ottmv{j_{{\mathrm{1}}}}  $ and $ \ottnt{k}  \Rightarrow^\ast   \ottmv{j_{{\mathrm{2}}}}  $.
\end{lemma}

\begin{lemma}[Inversion (l.r.)] \thmref{semantics.lean}{interpsStepInv} \label{lem:lr:inv} \\
  If $ \mathopen{\llbracket}  \ottnt{C}  \mathclose{\rrbracket}_{ \ottmv{i} } \searrow  \ottnt{P} $, then one of the following holds: $ \ottnt{C}  \Rightarrow^\ast   \bot  $; or
  \begin{itemize}[topsep=0pt]
    \item There exist $\ottnt{A}$ and $\ottnt{B}$ such that $ \ottnt{C}  \Rightarrow^\ast   \Pi  \ottmv{x}  \mathbin{:}  \ottnt{A}  \mathpunct{.}  \ottnt{B}  $; or
    \item There exists $\ottmv{i}$ such that $ \ottnt{C}  \Rightarrow^\ast   \kw{U} \gap   \ottmv{i}   $ or $ \ottnt{C}  \Rightarrow^\ast   \kw{Level}\texttt{<} \gap   \ottmv{i}   $.
  \end{itemize}
\end{lemma}

\subsection{Fundamental soundness theorem}

Although the logical relation relates closed types to sets of closed terms,
the fundamental theorem is proven over syntactic typing of open terms,
so we need a notion of semantic typing that handles closing over the terms
in a given typing context with a simultaneous substitution.
Semantic typing is then elementhood of a term in the interpretation of its type
for any substitution that closes them both.

At this point, referring to simultaneous substitutions is inevitable.
We denote them as $\sigma$, and write $ \sigma ,  \ottmv{x}  \mapsto  \ottnt{a} $
for its extension by a single substitution of $\ottmv{x}$ by $\ottnt{a}$.
In the mechanization, semantic well-typedness of a substitution \fbox{$ \sigma  \vDash  \Gamma $}
is defined similarly to semantic typing \fbox{$ \Gamma  \vdash  \ottnt{a}  \mathrel{:}  \ottnt{A} $},
but the admissible rules defined in \cref{fig:sem:subst} are more convenient.

\begin{definition} \thmref{semantics.lean}{semSubst}
  A substitution $\sigma$ is semantically well typed
  wrt context $\Gamma$ iff for every $ \ottmv{x}  \mathrel{:}  \ottnt{A}  \in  \Gamma $,
  there exist $\ottmv{i}, \ottnt{P}$ such that
  $ \mathopen{\llbracket}   \ottnt{A} [  \sigma  ]   \mathclose{\rrbracket}_{ \ottmv{i} } \searrow  \ottnt{P} $ and $   \mathit{ \ottmv{x} }  [  \sigma  ]   \in  \ottnt{P} $.
\end{definition}

\begin{figure}
\begin{mathpar}
  \mprset{fraction={\cdot\cdots\cdot}}
  \inferrule*[right=\ottdrulename{I-Nil}]{~}{ \sigma  \vDash   \cdot  }
  \and
  \inferrule*[right=\ottdrulename{I-Cons}]
    { \sigma  \vDash  \Gamma  \and
      \mathopen{\llbracket}   \ottnt{A} [  \sigma  ]   \mathclose{\rrbracket}_{ \ottmv{i} } \searrow  \ottnt{P}  \and
      \ottnt{a}  \in  \ottnt{P} }
    {  \sigma ,  \ottmv{x}  \mapsto  \ottnt{a}   \vDash   \Gamma ,  \ottmv{x}  \mathbin{:}  \ottnt{A}  }
\end{mathpar}
\caption{Semantically well-typed substitutions \thmref{semantics.lean}{semSubst\{Nil,Cons\}}}
\label{fig:sem:subst}
\end{figure}

\begin{definition}[Semantic typing] \thmref{semantics.lean}{semWt}
  A term $\ottnt{a}$ is semantically well typed with type $\ottnt{A}$ under context $\Gamma$,
  written $ \Gamma  \vDash  \ottnt{a}  :  \ottnt{A} $, iff for every $\sigma$ such that $ \sigma  \vDash  \Gamma $,
  there exist $\ottmv{i}, \ottnt{P}$ such that
  $ \mathopen{\llbracket}   \ottnt{A} [  \sigma  ]   \mathclose{\rrbracket}_{ \ottmv{i} } \searrow  \ottnt{P} $ and $  \ottnt{a} [  \sigma  ]   \in  \ottnt{P} $.
\end{definition}

The fundamental soundness theorem
states that syntactic typing implies semantic typing.
The cases corresponding to the rules in \cref{fig:typing:basic} are routine
by construction and inversion of \rref{I-Pi,I-Mty}~\citep{lr-pearl},
so we do not cover them all here.
Instead, we detail only the \rref*{I-Lam} case to highlight
where some of the above lemmas are used,
followed by the cases for the rules in \cref{fig:typing:univ}
that are unique to our system.
For concision, we skip steps involving massaging substitutions into the right shape.

\begin{theorem}[Soundness] \thmref{soundness.lean}{soundness} \label{thm:soundness}
  If $ \Gamma  \vdash  \ottnt{a}  \mathrel{:}  \ottnt{A} $, then $ \Gamma  \vDash  \ottnt{a}  :  \ottnt{A} $.
\end{theorem}

\begin{proof}
  By induction on the typing derivation.
  In each case, we suppose that $ \sigma  \vDash  \Gamma $.
  \begin{itemize}[topsep=0pt]
    \item \textit{\Rref{Lam}.}
      The relevant premises are $ \Gamma  \vdash   \Pi  \ottmv{x}  \mathbin{:}  \ottnt{A}  \mathpunct{.}  \ottnt{B}   \mathrel{:}   \kw{U} \gap  \ottnt{k}  $ and $  \Gamma ,  \ottmv{x}  \mathbin{:}  \ottnt{A}   \vdash  \ottnt{b}  \mathrel{:}  \ottnt{B} $,
      concluding with $ \Gamma  \vdash   \lambda  \ottmv{x}  \mathbin{:}  \ottnt{A}  \mathpunct{.}  \ottnt{b}   \mathrel{:}   \Pi  \ottmv{x}  \mathbin{:}  \ottnt{A}  \mathpunct{.}  \ottnt{B}  $.
      By the induction hypothesis on the first premise,
      \cref{lem:lr:inv-univ}, and \cref{lem:lr:inv-pi},
      we have $ \mathopen{\llbracket}   \ottnt{A} [  \sigma  ]   \mathclose{\rrbracket}_{ \ottmv{i} } \searrow  \ottnt{P_{{\mathrm{1}}}} $, $ \mathopen{\llbracket}   \ottnt{B} [   \sigma ,  \ottmv{x}  \mapsto  \ottnt{a}   ]   \mathclose{\rrbracket}_{ \ottmv{i} } \searrow  \ottnt{P_{{\mathrm{2}}}} $, and $ \ottnt{a}  \in  \ottnt{P_{{\mathrm{1}}}} $,
      where the goal is now to show that $   (  \lambda  \ottmv{x}  \mathbin{:}  \ottnt{A}  \mathpunct{.}  \ottnt{b}  )   \gap  \ottnt{a}   \in  \ottnt{P_{{\mathrm{2}}}} $.
      By \rref{I-Cons} and the induction hypothesis on the second premise,
      we have $ \mathopen{\llbracket}   \ottnt{B} [   \sigma ,  \ottmv{x}  \mapsto  \ottnt{a}   ]   \mathclose{\rrbracket}_{ \ottmv{i'} } \searrow  \ottnt{P'_{{\mathrm{2}}}} $ and $  \ottnt{b} [  \ottmv{x}  \mapsto  \ottnt{a}  ]   \in  \ottnt{P'_{{\mathrm{2}}}} $ for some $\ottmv{i'}, \ottnt{P'_{{\mathrm{2}}}}$.
      By \nameref{lem:lr:func}, we have that $\ottnt{P_{{\mathrm{2}}}} = \ottnt{P'_{{\mathrm{2}}}}$.
      Finally, by \nameref{lem:lr:back} on \rref{P-Beta} and $  \ottnt{b} [  \ottmv{x}  \mapsto  \ottnt{a}  ]   \in  \ottnt{P_{{\mathrm{2}}}} $,
      we obtain $   (  \lambda  \ottmv{x}  \mathbin{:}  \ottnt{A}  \mathpunct{.}  \ottnt{b}  )   \gap  \ottnt{a}   \in  \ottnt{P_{{\mathrm{2}}}} $.
    \item \textit{\Rref{Univ}.}
      The premise is $ \Gamma  \vdash  \ottnt{k}  \mathrel{:}   \kw{Level}\texttt{<} \gap  \ell  $,
      concluding with $ \Gamma  \vdash   \kw{U} \gap  \ottnt{k}   \mathrel{:}   \kw{U} \gap  \ell  $.
      By the induction hypothesis and \cref{lem:lr:inv-lvl},
      we have $ \ottmv{i_{{\mathrm{1}}}}  <  \ottmv{i_{{\mathrm{2}}}} $ such that $  \ottnt{k} [  \sigma  ]   \Rightarrow^\ast   \ottmv{i_{{\mathrm{1}}}}  $ and $  \ell [  \sigma  ]   \Rightarrow^\ast   \ottmv{i_{{\mathrm{2}}}}  $.
      By cofinality, there must exist a $\ottmv{j}$ such that $ \ottmv{i_{{\mathrm{2}}}}  <  \ottmv{j} $.
      The goal is now to show that $ \mathopen{\llbracket}   \kw{U} \gap   (  \ell [  \sigma  ]  )    \mathclose{\rrbracket}_{ \ottmv{j} } \searrow   \lbrace  \ottmv{z}  \mid   \exists  \ottnt{P}  \mathpunct{.}   \mathopen{\llbracket}   \mathit{ \ottmv{z} }   \mathclose{\rrbracket}_{ \ottmv{i_{{\mathrm{2}}}} } \searrow  \ottnt{P}    \rbrace  $
      and $ \mathopen{\llbracket}   \kw{U} \gap   (  \ottnt{k} [  \sigma  ]  )    \mathclose{\rrbracket}_{ \ottmv{i_{{\mathrm{2}}}} } \searrow   \lbrace  \ottmv{z}  \mid   \exists  \ottnt{P}  \mathpunct{.}   \mathopen{\llbracket}   \mathit{ \ottmv{z} }   \mathclose{\rrbracket}_{ \ottmv{i_{{\mathrm{1}}}} } \searrow  \ottnt{P}    \rbrace  $.
      These are both constructed using \rref{I-Univ} and \cref{lem:lr:pars}.
    \item \textit{\Rref{Level<}.}
      The premises are $ \Gamma  \vdash   \kw{U} \gap  \ottnt{k_{{\mathrm{1}}}}   \mathrel{:}   \kw{U} \gap  \ell_{{\mathrm{1}}}  $ and $ \Gamma  \vdash  \ottnt{k_{{\mathrm{0}}}}  \mathrel{:}   \kw{Level}\texttt{<} \gap  \ell_{{\mathrm{0}}}  $,
      concluding with $ \Gamma  \vdash   \kw{Level}\texttt{<} \gap  \ottnt{k_{{\mathrm{0}}}}   \mathrel{:}   \kw{U} \gap  \ottnt{k_{{\mathrm{1}}}}  $.
      By the induction hypothesis on the first premise and \cref{lem:lr:inv-univ},
      $ \kw{U} \gap   (  \ottnt{k_{{\mathrm{1}}}} [  \sigma  ]  )  $ has an interpretation as a universe,
      so it remains to find a $\ottnt{P}$ such that $ \mathopen{\llbracket}   \kw{Level}\texttt{<} \gap   (  \ottnt{k_{{\mathrm{0}}}} [  \sigma  ]  )    \mathclose{\rrbracket}_{ \ottmv{j} } \searrow  \ottnt{P} $,
      where $  \ottnt{k_{{\mathrm{1}}}} [  \sigma  ]   \Rightarrow^\ast   \ottmv{j}  $.
      By the induction on the second premise and \cref{lem:lr:inv-lvl},
      we have $  \ottnt{k_{{\mathrm{0}}}} [  \sigma  ]   \Rightarrow^\ast   \ottmv{i}  $ for some $\ottmv{i}$.
      Then the goal is constructed using \rref{I-Level<} and \cref{lem:lr:pars}.
    \item \textit{\Rref{Lvl}.} Straightforward by construction using \rref{I-Level<}.
    \item \textit{\Rref{Trans}.}
      The premises are $ \Gamma  \vdash  \ottnt{k_{{\mathrm{1}}}}  \mathrel{:}   \kw{Level}\texttt{<} \gap  \ottnt{k_{{\mathrm{2}}}}  $ and $ \Gamma  \vdash  \ottnt{k_{{\mathrm{2}}}}  \mathrel{:}   \kw{Level}\texttt{<} \gap  \ottnt{k_{{\mathrm{3}}}}  $,
      concluding with $ \Gamma  \vdash  \ottnt{k_{{\mathrm{1}}}}  \mathrel{:}   \kw{Level}\texttt{<} \gap  \ottnt{k_{{\mathrm{3}}}}  $.
      By the induction hypotheses on the two premises and \cref{lem:lr:inv-lvl},
      we know that $  \ottnt{k_{{\mathrm{1}}}} [  \sigma  ]   \Rightarrow^\ast   \ottmv{i_{{\mathrm{1}}}}  $, $  \ottnt{k_{{\mathrm{2}}}} [  \sigma  ]   \Rightarrow^\ast   \ottmv{i_{{\mathrm{2}}}}  $, $  \ottnt{k_{{\mathrm{2}}}} [  \sigma  ]   \Rightarrow^\ast   \ottmv{i'_{{\mathrm{2}}}}  $, and $  \ottnt{k_{{\mathrm{3}}}} [  \sigma  ]   \Rightarrow^\ast   \ottmv{i_{{\mathrm{3}}}}  $
      such that $ \ottmv{i_{{\mathrm{1}}}}  <  \ottmv{i_{{\mathrm{2}}}} $ and $ \ottmv{i'_{{\mathrm{2}}}}  <  \ottmv{i_{{\mathrm{3}}}} $.
      By \nameref{lem:par:confl} and \nameref{lem:par:consistency},
      it must be that $\ottmv{i_{{\mathrm{2}}}} = \ottmv{i'_{{\mathrm{2}}}}$.
      From the second inversion, we already know that $ \kw{Level}\texttt{<} \gap   (  \ottnt{k_{{\mathrm{3}}}} [  \sigma  ]  )  $ has an interpretation,
      so it remains to show that $  \ottnt{k_{{\mathrm{1}}}} [  \sigma  ]   \Rightarrow^\ast   \ottmv{i_{{\mathrm{1}}}}  $ and $  \ottnt{k_{{\mathrm{3}}}} [  \sigma  ]   \Rightarrow^\ast   \ottmv{i_{{\mathrm{3}}}}  $ such that $ \ottmv{i_{{\mathrm{1}}}}  <  \ottmv{i_{{\mathrm{3}}}} $,
      which holds by transitivity.
    \item \textit{\Rref{Cumul}.}
      The premises are $ \Gamma  \vdash  \ottnt{A}  \mathrel{:}   \kw{U} \gap  \ottnt{k}  $ and $ \Gamma  \vdash  \ottnt{k}  \mathrel{:}   \kw{Level}\texttt{<} \gap  \ell  $,
      concluding with $ \Gamma  \vdash  \ottnt{A}  \mathrel{:}   \kw{U} \gap  \ell  $.
      By induction on the first premise and \cref{lem:lr:inv-univ},
      we have some $\ottnt{P}$ such that $ \mathopen{\llbracket}   \ottnt{A} [  \sigma  ]   \mathclose{\rrbracket}_{ \ottmv{i} } \searrow  \ottnt{P} $ and $  \ottnt{k} [  \sigma  ]   \Rightarrow^\ast   \ottmv{i}  $.
      By induction on the second premise and \cref{lem:lr:inv-lvl},
      we have some $ \ottmv{i'}  <  \ottmv{j} $ such that $  \ottnt{k} [  \sigma  ]   \Rightarrow^\ast   \ottmv{i'}  $ and $  \ell [  \sigma  ]   \Rightarrow^\ast   \ottmv{j}  $.
      By \nameref{lem:par:confl} and \nameref{lem:par:consistency},
      it must be that $\ottmv{i} = \ottmv{i'}$.
      By cofinality and \cref{lem:lr:pars},
      $ \kw{U} \gap   (  \ell [  \sigma  ]  )  $ has an interpretation as a universe.
      It remains to show that $ \mathopen{\llbracket}   \ottnt{A} [  \sigma  ]   \mathclose{\rrbracket}_{ \ottmv{j} } \searrow  \ottnt{P} $,
      which holds by \nameref{lem:lr:cumul} on $ \ottmv{i}  <  \ottmv{j} $.
      \qedhere
  \end{itemize}
\end{proof}

Consistency and canonicity results then follow from the fundamental theorem as corollaries.

\begin{corollary}[Consistency] \thmref{soundness.lean}{consistency}
  There is no $\ottnt{b}$ such that $  \cdot   \vdash  \ottnt{b}  \mathrel{:}   \bot  $ holds.
  If there were, by \nameref{thm:soundness},
  we get have $  \cdot   \vDash  \ottnt{b}  :   \bot  $.
  Instantiating with the identity substitution,
  then inverting on the interpretation of $ \bot $,
  we get $ \ottnt{b}  \in   \varnothing  $, which is a contradiction.
\end{corollary}

\begin{corollary}[Canonicity of types] \thmref{soundness.lean}{canon$\mathcal{U}$} \label{lem:canon:univ}
  If $  \cdot   \vdash  \ottnt{C}  \mathrel{:}   \kw{U} \gap  \ottnt{k}  $,
  then $ \ottnt{C}  \Rightarrow^\ast   \Pi  \ottmv{x}  \mathbin{:}  \ottnt{A}  \mathpunct{.}  \ottnt{B}  $, $ \ottnt{C}  \Rightarrow^\ast   \kw{U} \gap   \ottmv{i}   $, $ \ottnt{C}  \Rightarrow^\ast   \kw{Level}\texttt{<} \gap   \ottmv{i}   $, or $ \ottnt{C}  \Rightarrow^\ast   \bot  $.
  By \nameref{thm:soundness},
  instantiating with the identity substitution,
  we have $\ottmv{j}, \ottnt{Q}$ such that $ \mathopen{\llbracket}   \kw{U} \gap  \ottnt{k}   \mathclose{\rrbracket}_{ \ottmv{j} } \searrow  \ottnt{Q} $ and $ \ottnt{C}  \in  \ottnt{Q} $.
  By inversion on the former,
  we have $\ottmv{i}, \ottnt{P}$ such that $ \ottnt{k}  \Rightarrow^\ast   \ottmv{i}  $ and $ \mathopen{\llbracket}  \ottnt{C}  \mathclose{\rrbracket}_{ \ottmv{i} } \searrow  \ottnt{P} $.
  Then the goal holds by \nameref{lem:lr:inv}.
\end{corollary}

\begin{corollary}[Canonicity of levels] \thmref{soundness.lean}{canonLvl}
  If $  \cdot   \vdash  \ottnt{k}  \mathrel{:}   \kw{Level}\texttt{<} \gap  \ell  $,
  then $ \ottnt{k}  \Rightarrow^\ast   \ottmv{i}  $ for some concrete level $\ottmv{i}$.
  By \nameref{thm:soundness},
  instantiating with the identity substitution,
  we have $\ottmv{j}, \ottnt{P}$ such that $ \mathopen{\llbracket}   \kw{Level}\texttt{<} \gap  \ell   \mathclose{\rrbracket}_{ \ottmv{j} } \searrow  \ottnt{P} $ and $ \ottnt{k}  \in  \ottnt{P} $.
  By inversion on the former,
  we have that $ \ell  \Rightarrow^\ast   \ottmv{i_{{\mathrm{2}}}}  $ and $ \ottnt{k}  \Rightarrow^\ast   \ottmv{i_{{\mathrm{1}}}}  $ such that $ \ottmv{i_{{\mathrm{1}}}}  <  \ottmv{i_{{\mathrm{2}}}} $.
\end{corollary}

\section{Towards normalization} \label{sec:normalization}

One conventional way to prove normalization,
given that we already have a syntactic logical relation,
is to extend it from closed to open types and terms.
However, we have not yet found the correct interpretation for open universe types
that continues to satisfy the same properties we need
(inversion, conversion, cumulativity, functionality)
while being strong enough for the soundness proof to go through.

It is also unclear whether the issue is finding the correct semantic model,
or if normalization does not hold at all,
because it depends on the syntactic presentation:
if we remove type annotations from our type theory
and present it Curry-style, is not normalizing.
While directly declaring an ill-founded level $ \ottmv{x}  \mathbin{:}   \kw{Level}\texttt{<} \gap   \mathit{ \ottmv{x} }   $ is impossible,
we can construct such a level in an inconsistent context
using an unannotated $ \kw{absurd} \gap   \relax  $ eliminator.
Then it becomes possible to type the universe at this level as its own type.
\Cref{fig:type-in-type} explicitly constructs the key part of the typing derivation
for $ \kw{U} \gap   (  \kw{absurd} \gap   \mathit{ \ottmv{x} }   )   :  \kw{U} \gap   (  \kw{absurd} \gap   \mathit{ \ottmv{x} }   )  $ where $\ottmv{x} :  \bot $.
With an instance of type-in-type,
we can construct a nonnormalizing lambda term via
\eg Hurkens' paradox~\citep{hurkens}.

\begin{figure}
\begin{mathpar}
  \inferrule*[Left=\rref*{Univ}]{
  \inferrule*[Left=\rref*{Abs}]{
    \inferrule*[Left=\rref*{Level<}]{
      \inferrule*[Left=\rref*{Univ}]{
      \inferrule*[Left=\rref*{Lvl}]{0 < 1}{  \ottmv{x}  \mathbin{:}   \bot    \vdash   0   \mathrel{:}   \kw{Level}\texttt{<} \gap   1   }}
      {  \ottmv{x}  \mathbin{:}   \bot    \vdash   \kw{U} \gap   0    \mathrel{:}   \kw{U} \gap   1   }
      \and
      \inferrule*[Right=\rref*{Abs}]{
      \inferrule*[]{\dots}{  \ottmv{x}  \mathbin{:}   \bot    \vdash   \kw{Level}\texttt{<} \gap   0    \mathrel{:}   \kw{U} \gap   0   }
      \and
      \inferrule*[Right=\rref*{Var}]{ \ottmv{x}  \mathrel{:}   \bot   \in   \ottmv{x}  \mathbin{:}   \bot   }{  \ottmv{x}  \mathbin{:}   \bot    \vdash   \mathit{ \ottmv{x} }   \mathrel{:}   \bot  }}
      {  \ottmv{x}  \mathbin{:}   \bot    \vdash   \kw{absurd} \gap   \mathit{ \ottmv{x} }    \mathrel{:}   \kw{Level}\texttt{<} \gap   0   }}
      {  \ottmv{x}  \mathbin{:}   \bot    \vdash   \kw{Level}\texttt{<} \gap   (  \kw{absurd} \gap   \mathit{ \ottmv{x} }   )    \mathrel{:}   \kw{U} \gap   0   }
    \and
    \inferrule*[Right=\rref*{Var}]{ \ottmv{x}  \mathrel{:}   \bot   \in   \ottmv{x}  \mathbin{:}   \bot   }{  \ottmv{x}  \mathbin{:}   \bot    \vdash   \mathit{ \ottmv{x} }   \mathrel{:}   \bot  }}
    {  \ottmv{x}  \mathbin{:}   \bot    \vdash   \kw{absurd} \gap   \mathit{ \ottmv{x} }    \mathrel{:}   \kw{Level}\texttt{<} \gap   (  \kw{absurd} \gap   \mathit{ \ottmv{x} }   )   }}
  {  \ottmv{x}  \mathbin{:}   \bot    \vdash   \kw{U} \gap   (  \kw{absurd} \gap   \mathit{ \ottmv{x} }   )    \mathrel{:}   \kw{U} \gap   (  \kw{absurd} \gap   \mathit{ \ottmv{x} }   )   }
\end{mathpar}
\caption{Type-in-type in an inconsistent context}
\label{fig:type-in-type}
\end{figure}

The ability to assign different types to the term $ \kw{absurd} \gap   \mathit{ \ottmv{x} }  $ is key to constructing this derivation.
By requiring a type annotation that gets compared during definitional equality,
we can only construct a derivation for
$ \kw{U} \gap   (  \kw{absurd}_{  (  \kw{Level}\texttt{<} \gap   (  \kw{absurd}_{  (  \kw{Level}\texttt{<} \gap   0   )  } \gap   \mathit{ \ottmv{x} }   )   )  } \gap   \mathit{ \ottmv{x} }   )   :  \kw{U} \gap   (  \kw{absurd}_{  (  \kw{Level}\texttt{<} \gap   0   )  } \gap   \mathit{ \ottmv{x} }   )  $,
which cannot be used as type-in-type.
For similar reasons, we cannot use $\ottmv{x} :  \Pi  A  \mathbin{:}   \kw{U} \gap   \mathit{ i }    \mathpunct{.}   \mathit{ A }  $
to construct the ill-founded level, as the type arguments will be incomparable.
In contrast, type annotations have no influence on consistency,
as it remains provable via the logical relation on closed types
even when annotations are removed.

\section{Extensions} \label{sec:extensions}

Our type theory is intentionally minimal to focus only on the core necessities
of first-class levels and to keep the proof development small and uncluttered.
Some reasonable extensions include the remaining missing types from MLTT,
\ie dependent pairs, sums, naturals, propositional equality, and W types,
or general inductive types as in CIC~\citep{pcuic}.
However, these features and their difficulties are orthogonal from universes and levels.
Here, we instead look at extensions that augment how universes and levels behave,
some of which are validated by our current semantics,
and others which present additional challenges.

\subsection{Level operators and eliminators}

The only features missing from \lang that Agda has are a zeroth level,
a level successor operator, and a level maximum operator.
To justify them semantically,
we would impose the first two as additional existence conditions on the metalevel levels;
the third follows from the total ordering, which lets us pick the larger of two levels.

\vspace{-0.75\baselineskip}
\begin{mathpar}
  \inferrule[\ottdrulename{Zero}]
    { \Gamma  \vdash  \ottnt{k}  \mathrel{:}   \kw{Level}\texttt{<} \gap  \ell  }
    { \Gamma  \vdash   0   \mathrel{:}   \kw{Level}\texttt{<} \gap   (  \mathop{\uparrow}  \ottnt{k}  )   }
  \and
  \inferrule[\ottdrulename{Succ}]
    { \Gamma  \vdash  \ottnt{k}  \mathrel{:}   \kw{Level}\texttt{<} \gap  \ell  }
    { \Gamma  \vdash   \mathop{\uparrow}  \ottnt{k}   \mathrel{:}   \kw{Level}\texttt{<} \gap   (  \mathop{\uparrow}  \ell  )   }
  \and
  \inferrule[\ottdrulename{Max}]
    { \Gamma  \vdash  \ottnt{k_{{\mathrm{1}}}}  \mathrel{:}   \kw{Level}\texttt{<} \gap  \ell_{{\mathrm{1}}}   \and
      \Gamma  \vdash  \ottnt{k_{{\mathrm{2}}}}  \mathrel{:}   \kw{Level}\texttt{<} \gap  \ell_{{\mathrm{2}}}  }
    { \Gamma  \vdash   \ottnt{k_{{\mathrm{1}}}}  \sqcup  \ottnt{k_{{\mathrm{2}}}}   \mathrel{:}   \kw{Level}\texttt{<} \gap   (  \ell_{{\mathrm{1}}}  \sqcup  \ell_{{\mathrm{2}}}  )   }
\end{mathpar}

What complicates matters are the additional definitional equalities that
ensure that the maximum operator is idempotent, associative, commutative,
distributive with respect to successors,
and that $0$ is its identity element.
While these properties hold automatically at the metalevel for concrete levels,
they do not for arbitrary level expressions,
\eg $   0   \sqcup   \mathop{\uparrow}   (   \mathit{ \ottmv{x} }   \sqcup   \mathop{\uparrow}   \mathit{ \ottmv{x} }    )     \equiv   \mathop{\uparrow}   \mathop{\uparrow}   \mathit{ \ottmv{x} }    $.
Our notions of reduction then need to pick a direction for each equality
to reduce levels to some chosen canonical form.
We believe the mechanization to be doable but tedious.

Meanwhile, well-founded induction on levels already holds semantically,
as we need it to define our logical relation in the first place.
We can internalize it by syntactically introducing an eliminator $ \kw{wf} \gap   \relax  $ for levels,
which states that a predicate $\ottnt{B}$ holds on arbitrary levels
if we can show that it holds for a given level
when we know it holds for all smaller levels.
However, it is unclear whether such an eliminator would be useful.

\vspace{-0.75\baselineskip}
\begin{mathpar}
  \inferrule[ElimLvl]
    {  \Gamma ,  \ottmv{z}  \mathbin{:}   \kw{Level}\texttt{<} \gap  \ottnt{k}    \vdash  \ottnt{B}  \mathrel{:}   \kw{U} \gap  \ell   \\\\
      \Gamma  \vdash  \ottnt{b}  \mathrel{:}     \Pi  \ottmv{x}  \mathbin{:}   \kw{Level}\texttt{<} \gap  \ottnt{k}   \mathpunct{.}   (   \Pi  \ottmv{y}  \mathbin{:}   \kw{Level}\texttt{<} \gap   \mathit{ \ottmv{x} }    \mathpunct{.}  \ottnt{B}  [  \ottmv{z}  \mapsto   \mathit{ \ottmv{y} }   ]  )    \to  \ottnt{B}  [  \ottmv{z}  \mapsto   \mathit{ \ottmv{x} }   ]  }
    { \Gamma  \vdash   \kw{wf} \gap  \ottnt{b}   \mathrel{:}   \Pi  \ottmv{z}  \mathbin{:}   \kw{Level}\texttt{<} \gap  \ottnt{k}   \mathpunct{.}  \ottnt{B}  }
  \and
  \inferrule[E-ElimLvl]{~}{   \kw{wf} \gap  \ottnt{b}   \gap  \ottnt{k}   \equiv    \ottnt{b}  \gap  \ottnt{k}   \gap   (   \lambda  \ottmv{y}  \mathpunct{.}   \kw{wf} \gap  \ottnt{b}    \gap   \mathit{ \ottmv{y} }   )   }
\end{mathpar}

\subsection{Subtyping} \label{sec:subtyping}

Because levels are now terms,
subtyping necessarily involves typing to compare two levels.
In particular, a universe at a smaller level is a subtype of a one at a larger level,
while a level type bounded by a smaller level is a subtype of a one bounded by a larger level.
The former is already expressed by \rref{Cumul}, the latter by \rref{Trans}.
The additional benefit of subtyping making
function domains contravariant and codomains covariant with respect to subtyping.
Selected subtyping rules are given below,
along with an updated \rref{Conv'} rule.

\vspace{-0.75\baselineskip}
\begin{mathpar}
  \inferrule[\ottdrulename{S-Univ}]
    { \Gamma  \vdash  \ottnt{k}  \mathrel{:}   \kw{Level}\texttt{<} \gap  \ell  }
    { \Gamma  \vdash   \kw{U} \gap  \ottnt{k}   \preccurlyeq   \kw{U} \gap  \ell  }
  \and
  \inferrule[\ottdrulename{S-Level<}]
    { \Gamma  \vdash  \ottnt{k}  \mathrel{:}   \kw{Level}\texttt{<} \gap  \ell  }
    { \Gamma  \vdash   \kw{Level}\texttt{<} \gap  \ottnt{k}   \preccurlyeq   \kw{Level}\texttt{<} \gap  \ell  }
  \and
  \inferrule[\ottdrulename{S-Pi}]
    { \Gamma  \vdash  \ottnt{A_{{\mathrm{2}}}}  \preccurlyeq  \ottnt{A_{{\mathrm{1}}}}  \and
       \Gamma ,  \ottmv{x}  \mathbin{:}  \ottnt{A_{{\mathrm{1}}}}   \vdash  \ottnt{B_{{\mathrm{1}}}}  \preccurlyeq  \ottnt{B_{{\mathrm{2}}}} }
    { \Gamma  \vdash   \Pi  \ottmv{x}  \mathbin{:}  \ottnt{A_{{\mathrm{1}}}}  \mathpunct{.}  \ottnt{B_{{\mathrm{1}}}}   \preccurlyeq   \Pi  \ottmv{x}  \mathbin{:}  \ottnt{A_{{\mathrm{2}}}}  \mathpunct{.}  \ottnt{B_{{\mathrm{2}}}}  }
  \and
  \inferrule[\ottdrulename{S-Trans}]
    { \Gamma  \vdash  \ottnt{A}  \preccurlyeq  \ottnt{B}  \and
      \Gamma  \vdash  \ottnt{B}  \preccurlyeq  \ottnt{C} }
    { \Gamma  \vdash  \ottnt{A}  \preccurlyeq  \ottnt{C} }
  \and
  \inferrule[\ottdrulename{S-Conv}]
    { \ottnt{A}  \equiv  \ottnt{B} }
    { \Gamma  \vdash  \ottnt{A}  \preccurlyeq  \ottnt{B} }
  \and
  \inferrule[\ottdrulename{Conv'}]
    { \Gamma  \vdash  \ottnt{a}  \mathrel{:}  \ottnt{A}  \and
      \Gamma  \vdash  \ottnt{B}  \mathrel{:}   \kw{U} \gap  \ottnt{k}   \and
      \Gamma  \vdash  \ottnt{A}  \preccurlyeq  \ottnt{B} }
    { \Gamma  \vdash  \ottnt{a}  \mathrel{:}  \ottnt{B} }
\end{mathpar}

Although all of this subtyping behaviour holds semantically in our current model,
proving logical consistency is not so easy.
The simplicity of our logical relation relies on
the independence of definitional equality from typing,
along with its equivalence to conversion.
By introducing a subtyping judgement that depends on typing,
which in turn depends on subtyping, to prove consistency,
the logical relation would need to include a semantic notion of equality,
similar to the reducibility judgements used by Abel, \"Ohman, and Vezzosi~\citep{dec-conv}.

\section{Conclusion and future work} \label{sec:conclusion}

We have presented \lang, a type theory with first-class universe levels.
In contrast to existing work,
rather than level constraints being separate from the type of levels,
we combine them such that every level explicitly has a bound.
We have proven our type theory to be type safe,
and in particular that subject reduction holds.
This is in contrast to BCDE~\citep{univ-poly},
the only other formal syntactic system we know of
with universe level polymorphism beyond prenex polymorphism,
which violates subject reduction.
We have also proven our type theory to be logically consistent,
and therefore useable as a logic for writing proofs.

Proving normalization and decidability of type checking
is the next step in showing that our type theory is effectively type checkable
and thus has the potential to be a basis for theorem proving.
Whether the extended logical relation presented in \cref{sec:normalization}
can be repaired to prove normalization is unclear,
as is whether well-typed terms are normalizing at all.
Looking to existing work,
BCDE proposes allowing looping level constraints of the form $\ottnt{k} < \ottnt{k}$ to admit subject reduction,
but this would also permit type-in-type in a looping context and violate normalization.
Even so, we are hopeful that it holds,
as no issues with cumulative first-class levels have yet arisen in Agda.

Decidability of type checking does not hold straightforwardly from normalization,
as a type checking algorithm must incorporate the non--syntax-directed \rref{Trans,Cumul}.
It may be done separately via algorithmic subtyping,
but as seen in \cref{sec:subtyping},
a subtyping relation must depend on typing
to show that one level expression is strictly smaller than another.
The challenge lies in showing totality of a mutual typing--subtyping algorithm,
but if looping level bounds $\ottnt{k} :  \kw{Level}\texttt{<} \gap  \ottnt{k} $ are ruled out by normalization,
there is no reason to believe it would not be total.

\bibliographystyle{plainurl}
\bibliography{main.bib}

\end{document}